\documentclass{imaiai-modified}

\usepackage{amsmath, amssymb, bm, cite, epsfig, psfrag}
\usepackage{algorithm,algorithmic}
\usepackage{ifthen}
\usepackage{amsthm}
\usepackage{enumerate}
\usepackage{epstopdf}

\def\beq{\begin{equation}}
\def\eeq{\end{equation}}
\def\beqa{\begin{eqnarray}}
\def\eeqa{\end{eqnarray}}
\def\beqan{\begin{eqnarray*}}
\def\eeqan{\end{eqnarray*}}

\def\R{{\mathbb{R}}}

\def\argmin{\mathop{\mathrm{arg\,min}}}
\def\argmax{\mathop{\mathrm{arg\,max}}}

\def\x{\times}

\def\MID{\,|\,}

\def\Ecal{{\mathcal E}}

\newtheoremstyle{assumption}{6pt}{6pt}{\rm}{}{\sffamily}{ }{ }{}
\theoremstyle{assumption}
\newtheorem{assumption}[theorem]{\sc Assumption}

\def\SNR{\mbox{\small \sffamily SNR}}

\def\arr{\rightarrow}
\def\Exp{\mathbb{E}}
\def\var{\mbox{\bf var}}
\def\tm1{t\! - \! 1}
\def\tp1{t\! + \! 1}
\def\m1{\! - \! 1}
\def\p1{\! + \! 1}

\newcommand{\bbf}{\mathbf{b}}
\newcommand{\dbf}{\mathbf{d}}

\newcommand{\pbf}{\mathbf{p}}

\newcommand{\qbf}{\mathbf{q}}

\newcommand{\ubf}{\mathbf{u}}
\newcommand{\ubfhat}{\widehat{\mathbf{u}}}

\newcommand{\vbf}{\mathbf{v}}
\newcommand{\vbfhat}{\widehat{\mathbf{v}}}

\newcommand{\xbf}{\mathbf{x}}

\newcommand{\ybf}{\mathbf{y}}

\newcommand{\zbf}{\mathbf{z}}

\newcommand{\Abf}{\mathbf{A}}

\newcommand{\Ibf}{\mathbf{I}}

\newcommand{\Sbf}{\mathbf{S}}

\newcommand{\Vbf}{\mathbf{V}}

\newcommand{\Wbf}{\mathbf{W}}

\def\alphabar{\overline{\alpha}}
\def\nubar{\overline{\nu}}
\def\lambdabar{\overline{\lambda}}

\begin{document}

\title{Iterative Reconstruction of Rank-One Matrices in Noise}

\shorttitle{Rank-One Reconstruction} 
\shortauthorlist{Fletcher, Rangan} 

\author{{\sc Alyson K.\ Fletcher}$^*$,\\[2pt]
University of California, Santa Cruz, CA 95064 \\
$^*${\email{akfletch@ucsc.edu}}\\[2pt]
{\sc and}\\[6pt]
{\sc Sundeep Rangan}\\[2pt]
New York University Polytechnic School of Engineering, Brooklyn, NY 11201\\
{srangan@nyu.edu}
}

\maketitle

\begin{abstract}
{We consider the problem of estimating a rank-one matrix in Gaussian noise
under a probabilistic model for the left and right factors of the matrix.
The probabilistic model can impose constraints on the factors including
sparsity and positivity that arise commonly in learning problems.
We propose a family of algorithms that reduce the problem to a sequence of
scalar estimation computations.  These algorithms are
similar
to approximate message passing techniques based on Gaussian approximations of loopy
belief propagation that have been used recently in compressed sensing.
Leveraging analysis methods by Bayati and Montanari, we
show that the asymptotic behavior of the algorithm is described
by a simple scalar equivalent model, where the distribution of the estimates at each iteration is identical
to certain scalar estimates of the variables
in Gaussian noise.  Moreover, the effective Gaussian noise
level is described by a set of state evolution equations.
The proposed approach to deriving algorithms
thus provides a computationally simple and general method
for rank-one estimation problems with a precise analysis in certain high-dimensional
settings.}
{Matrix factorization, Bayesian estimation, approximate message passing}
\\
2000 Math Subject Classification: 34K30, 35K57, 35Q80,  92D25
\end{abstract}

\section{Introduction} \label{sec:intro}
We consider the problem
of estimating vectors $\ubf_0 \in \R^m$ and $\vbf_0 \in\R^n$
from a matrix $\Abf\in \R^{m \x n}$ of the form
\beq \label{eq:ArankOne}
    \Abf = \ubf_0\vbf^T_0 + \sqrt{m}\Wbf,
\eeq
where $\Wbf$ represents some unknown noise and $\sqrt{m}$ is a normalization factor.
The problem can be considered as a rank-one special case of finding a low-rank
matrix in the presence of noise.  Such low-rank estimation
problems arise in a range
of applications including blind channel estimation 
\citep{Ding:97},
antenna array processing \citep{JafGold:04},
subspace system identification \citep{Katayama:10},
and principal component or factor analysis \citep{Jolliffe:86}.

When the noise term $\Wbf$ is zero, the vector pair $(\ubf_0,\vbf_0)$
can be recovered exactly, up to a scaling, from the maximal
left and right singular vectors
of $\Abf$ \citep{HornJ:85}.
However, in the presence of noise, the rank-one matrix can in general only be estimated
approximately.  In this case, \emph{a priori} information or
constraints on $(\ubf_0,\vbf_0)$ may improve the estimation.
Such constraints arise, for example, in factor analysis in statistics,
where one of the factors is often constrained to be either positive or sparse
\citep{CadJoll:95}.  Similar sparsity constraints occur in the
problem of dictionary learning \citep{OlshausenF:96}.
Also, in digital communications,
one of the factors could come from a discrete QAM constellation.

In this paper, we impose the constraints on $\ubf_0$ and $\vbf_0$
in a Bayesian setting where $\ubf_0$ and $\vbf_0$ are assumed to be
independent from one another with i.i.d.\ densities of the form
\beq \label{eq:puvDist}
    p_{\ubf_0}(\ubf_0) = \prod_{i=1}^m p_{U_0}(u_{0i}), \quad
    p_{\vbf_0}(\vbf_0) = \prod_{j=1}^n p_{V_0}(v_{0j}),
\eeq
for some scalar random variables $U_0$ and $V_0$.  The
noise $\Wbf$
in \eqref{eq:ArankOne} is assumed to have i.i.d.\ Gaussian components where
$W_{ij} \sim {\mathcal N}(0,\tau_w)$ for some variance $\tau_w > 0$.
Under this model, the posterior density of $\ubf_0$ and $\vbf_0$
is given by
\beq \label{eq:puvA}
    p_{\ubf_0,\vbf_0|\Abf}(\ubf_0,\vbf_0) \propto p_{\ubf_0}(\ubf_0)p_{\vbf_0}(\vbf_0)
    \exp\left[ -\frac{1}{2m\tau_w}\|\Abf-\ubf_0\vbf^T\|^2_F \right],
\eeq
where $\|\cdot\|_F$ is the Frobenius norm.
Given this posterior density, we consider two estimation problems:
\begin{itemize}
\item \emph{MAP estimation:}  In this problem, we wish to find the
\emph{maximum a posteriori} estimates
\beq \label{eq:uvMAP}
    (\ubfhat,\vbfhat) = \argmax_{\ubf,\vbf} p_{\ubf_0,\vbf_0|\Abf}(\ubf,\vbf).
\eeq
\item \emph{MMSE estimation:} In this problem, we wish to find the posterior
mean or, equivalently, the minimum mean squared error estimate,
\beq \label{eq:uvMMSE}
    \ubfhat = \Exp(\ubf_0|\Abf), \quad \vbfhat = \Exp(\vbf_0|\Abf).
\eeq
We may also be interested in the posterior variances and posterior marginals.
\end{itemize}

Exact computation of either of these estimates is generally
intractable.
Even though the components $\ubf_0$ and $\vbf_0$ are assumed to
have independent components,
the posterior density \eqref{eq:puvA} is not, in general, separable due to the
term $\ubf_0\vbf_0^T$.
Thus, the MAP estimate requires a search over $m$ and $n$-dimensional vectors,
$(\ubf_0,\vbf_0)$ and the MMSE estimate requires integration over this $m+n$-dimensional
space.
However, due to the separability assumption on the priors \eqref{eq:puvDist},
it is computationally simpler to alternately estimate the factors $\ubf_0$ and $\vbf_0$.
This form of alternating estimation is used, for example,
in the classic alternating power method for finding maximal singular
values \citep{HornJ:85}
and some alternating methods in sparse or non-negative dictionary learning
\citep{OlshausenF:96,OlshausenF:97,LeeSeung:01,LewickiS:00,AharonEB:06}.

The approach in this work also uses a form of alternating
estimation, but based on a recently-developed powerful class of
algorithms known as Approximate Message Passing or AMP.
AMP methods are derived from Gaussian and quadratic approximations
of loopy belief propagation in graphical models
and were originally used for problems in compressed sensing
\citep{DonohoMM:09,DonohoMM:10-ITW1,DonohoMM:10-ITW2,BayatiM:11,Rangan:10arXiv,Rangan:11-ISIT}.
The AMP methodology has been successfully applied in a range of applications
\citep{FletcherRVB:11,chen2010improved,vila2013hyperspectral,fletcher2014scalable}.
In recent years, there has been growing interest of AMP for
matrix factorization problems as well
\citep{RanganF:12-ISIT,parker2013bilinear,parker2013bilinear2,krzakala2013phase,lesieur2015mmse}.
The problem considered in this paper can be considered as a simple rank one
case of these problems.

Our main contribution in this work is to show that, for the rank one MMSE and MAP estimation
problems,
the proposed IterFac algorithm admits an asymptotically-exact characterization
in the limit when $m,n \arr \infty$ with $m/n$ constant
and the components of the true
vectors $\ubf_0$ and $\vbf_0$ have limiting empirical distributions.
In this scenario, we show that the empirical joint distribution of the components
of $\ubf_0$ and $\vbf_0$
and the corresponding estimates from the IterFac algorithm
are described by a simple \emph{scalar equivalent model}
where the IterFac component estimates are identically distributed to scalar estimates
of the variables corrupted by Gaussian noise. Moreover, the effective Gaussian
noise level in this model
is described by a simple set of scalar \emph{state evolution} (SE) equations.
This scalar equivalent model is common in analyses of AMP methods
\citep{DonohoMM:09,DonohoMM:10-ITW1,DonohoMM:10-ITW2,BayatiM:11,Rangan:10arXiv,Rangan:11-ISIT}
as well as replica analyses of related estimation problems \citep{Tanaka:02,GuoV:05,RanganFG:12-IT}.
From the scalar equivalent model, one can
compute the asymptotic value of almost any component-separable metric
including
mean-squared error or correlation.
Thus, in addition to being computationally simple and general,
the IterFac algorithm admits a precise analysis in the case of Gaussian noise.
Moreover, since fixed points of the IterFac algorithm correspond, under
suitable circumstances, to local maxima of objectives such as \eqref{eq:uvMAP},
the analysis can be used to characterize the behavior of such minima---even
if alternate algorithms to IterFac are used.

The main analytical tool is a recently-developed technique
by Bayati and Montanari \citep{BayatiM:11}
used in the analysis of AMP algorithms.
This work proved that, in the limit for large Gaussian mixing matrices, the behavior of
AMP estimates can be described by
a scalar equivalent model with effective noise levels defined by certain
scalar state evolution (SE) equations.
Similar SE analyses have appeared in related contexts
\citep{BouCaire:02,MontanariT:06,GuoW:06,GuoW:07,Rangan:10arXiv,Rangan:11-ISIT}.
To prove the SE equations for the IterFac algorithm,
we apply a key theorem from \citep{BayatiM:11}
with a simple change of variables and a slight modification to account
for parameter adaptation.

A conference version of this paper appeared in \citep{RanganF:12-ISIT}.
This paper provides all the proofs along with more detailed discussions
and simulations.
Since the original publication of the conference version in \citep{RanganF:12-ISIT},
several other research groups have extended and built on the work.
Importantly, \citep{deshpande2014information} has shown that in the case
of certain discrete priors, the IterFac algorithm is provably Bayes optimal
in a large system limit.
It was shown in \citep{deshpande2015finding}, that
an algorithm closely related to IterFac  could provide the best known
scaling laws for the hidden clique problem.
More recently, \citep{parker2014bilinear} and
\citep{kabashima2014phase} have used related AMP-type methods for
matrix factorization problems with rank greater than one.

\section{Iterative Rank-One Factorization}

\begin{algorithm}
\caption{Iterative Factorization (IterFac) }
\begin{algorithmic}[1]  \label{algo:iterMax}
\REQUIRE{ Matrix $\Abf \in \R^{m \x n}$, noise level $\tau_w > 0$,
and factor selection functions $G_u(\pbf,\lambda_u)$ and
$G_v(\qbf,\lambda_v)$.}
\STATE{ $t \gets 0$ }
\STATE{ $\mu_u(t) \gets 0$ and select initial values $\ubf(0)$, $\vbf(0)$ }
\REPEAT

\STATE \COMMENT{Update estimate of $\ubf$}
\STATE{Select $\lambda_u(t)$} \label{line:lamu}
\STATE{$\pbf(t) \gets (1/m)\Abf\vbf(t) + \mu_u(t)\ubf(t)$ }
    \label{line:pt}
\STATE{$\ubf(\tp1) \gets G_u(\pbf(t),\lambda_u(t))$} \label{line:Guopt}
\STATE{$\mu_v(t) \gets -(\tau_w/m)\sum_{i=1}^m \partial G_u(p_i(t),\lambda_u(t))/\partial p_i$} \label{line:muu}
\STATE{$~$}

\COMMENT{Update estimate of $\vbf$}
\STATE{Select $\lambda_v(t)$} \label{line:lamv}
\STATE{$\qbf(t) \gets (1/m)\Abf^T\ubf(\tp1) + \mu_v(t)\vbf(t)$ }
    \label{line:qt}
\STATE{$\vbf(\tp1) \gets G_v(\qbf(t),\lambda_v(t))$} \label{line:Gvopt}
\STATE{$\mu_u(\tp1) \gets -(\tau_w/m)\sum_{j=1}^n \partial G_u(q_j(t),\lambda_u(t))/\partial q_j$}
        \label{line:muv}

\UNTIL{Terminate}

\end{algorithmic}
\end{algorithm}

The proposed IterFac algorithm is shown in Algorithm~\ref{algo:iterMax}.
Given a matrix $\Abf \in \R^{m\x n}$, the algorithm outputs
a sequence of estimates
$(\ubf(t),\vbf(t))$, $t=0,1,\ldots$, for $(\ubf_0,\vbf_0)$.
The algorithm has several parameters
including the initial conditions,
the parameters in lines \ref{line:lamu}
and \ref{line:lamv}, the termination condition and, most importantly,
the functions $G_u(\cdot)$ and $G_v(\cdot)$.
In each iteration, the functions $G_u(\cdot)$ and $G_v(\cdot)$
are used to generate the estimates of the factors
$\ubf(t)$ and $\vbf(t)$ and will be called the \emph{factor selection functions}.
Throughout this work,
we assume that the factor selection functions are \emph{separable} meaning that
the act on the componentwise on the vectors $\pbf(t)$ and $\qbf(t)$:
\beq \label{eq:Guvsep}
    u_i(t) = G_u(p_i(t),\lambda_u(t)), \quad
    v_j(\tp1) = G_v(q_j(t),\lambda_v(t)),
\eeq
for some scalar functions $G_u(\cdot)$ and $G_v(\cdot)$.
The choice of the factor selection function will depend on whether IterFac is used
for MAP or MMSE estimation.

\paragraph*{MAP estimation:}
To describe the factor selection functions for the MAP estimation problem
\eqref{eq:uvMAP}, Let $\lambda_u(t)$ and $\lambda_v(t)$ be sequences of pairs
of parameters
\beq \label{eq:lamuvDef}
    \lambda_u(t) := (\gamma_u(t),\nu_u(t)), \quad
    \lambda_v(t) := (\gamma_v(t),\nu_v(t)).
\eeq
Then, for each $t$, consider random vectors $\pbf(t)$ and $\qbf(t)$ given by,
\begin{align} \label{eq:pqAwgn}
    \begin{split}
        &\pbf(t) = \gamma_u(t)\ubf_0 + \zbf_u(t), \quad \ubf_0 \sim p_{\ubf_0}(\ubf_0),
            \quad \zbf_u(t) \sim {\mathcal N}(0,\nu_u(t)\Ibf) \\
        &\qbf(t) = \gamma_v(t)\vbf_0 + \zbf_v(t), \quad \vbf_0 \sim p_{\vbf_0}(\vbf_0),
            \quad \zbf_v(t) \sim {\mathcal N}(0,\nu_v(t)\Ibf).
    \end{split}
\end{align}
The random variables $\pbf(t)$ and $\qbf(t)$ in \eqref{eq:pqAwgn}
are simply scaled versions of the true vectors $\ubf_0$ and $\vbf_0$ with additive white
Gaussian noise (AWGN).
Then, for the MAP problem we take the
factor selection functions to be the MAP estimates of $\ubf_0$ and $\vbf_0$
given the observations $\pbf(t)$ and $\qbf(t)$:
\beq \label{eq:GuvMAP}
        G_u(\pbf(t),\lambda_u(t)) := \argmax_{\ubf_0}
            p(\ubf_0|\pbf(t),\gamma_u(t),\nu_u(t)), \quad
        G_v(\qbf(t),\lambda_v(t)) = \argmax_{\vbf_0}
            p(\vbf_0|\qbf(t),\gamma_v(t),\nu_v(t)).
\eeq
Importantly, due to the separability assumption on the priors \eqref{eq:puvDist},
these MAP estimates can be computed componentwise:
\begin{align} \label{eq:GuvMAPsep}
\begin{split}
    & G_u(p_i,\lambda_u) = \argmin_{u_{0i}} \left[ -\log p_{U_0}(u_{0i})
        + \frac{(\gamma_uu_{0i}-p_i)^2}{2\nu_u} \right], \\
     & G_v(q_j,\lambda_v) = \argmin_{v_{0j}} \left[ -\log p_{V_0}(v_{0j})
        + \frac{(\gamma_vv_{0j}-q_j)^2}{2\nu_v} \right].
\end{split}
\end{align}
Hence, the IterFac algorithm replaces the vector-valued MAP estimation of $\ubf_0,\vbf_0$
from the joint density \eqref{eq:puvA}, with a sequence of scalar MAP estimation
problems along with multiplications by $\Abf$ and $\Abf^T$.

\paragraph*{MMSE Estimation:}  For the MMSE estimation problem, we simply take
the factor selection functions
the MMSE estimates with respect to the random variables \eqref{eq:pqAwgn},
\beq \label{eq:GuvMMSE}
        G_u(\pbf(t),\lambda_u(t)) = \Exp(\ubf_0|\pbf(t),\gamma_u(t),\nu_u(t)), \quad
        G_v(\qbf(t),\lambda_v(t)) = \Exp(\vbf_0|\qbf(t),\gamma_v(t),\nu_v(t)).
\eeq
Again, due to the separability assumption \eqref{eq:puvDist}, these
MMSE estimation problems can be performed componentwise.  Hence MMSE IterFac
reduces the vector MMSE problem to a sequence of scalar problems in Gaussian noise.

\section{Intuitive Algorithm Derivation and Analysis}

\subsection{Algorithm Intuition}
Before we formally analyze the algorithm, it is instructive to understand
the intuition behind the method.
For both the MAP and MMSE versions of IterFac, first observe that
\beq \label{eq:pintuit}
    \pbf(t) \stackrel{(a)}{=} \frac{1}{m}\Abf\vbf(t)+\mu_u(t)\ubf(t)
    \stackrel{(b)}{=} \frac{\vbf_0^T\vbf(t)}{m} \ubf_0
        +  \frac{1}{\sqrt{m}}\Wbf\vbf(t) + \mu_u(t)\ubf(t)
        \stackrel{(c)}{=} \gamma_u(t)\ubf_0 + \zbf_u(t),
\eeq
where (a) follows from line~\ref{line:pt} in Algorithm~\ref{algo:iterMax};
(b) follows from
the assumptions on the measurements \eqref{eq:ArankOne} and in (c), we have used
the definitions
\beq \label{eq:gamzu}
    \gamma_u(t) = \frac{1}{m}\vbf_0^T\vbf(t), \quad
    \zbf_u(t) = \frac{1}{\sqrt{m}}\Wbf\vbf(t) + \mu_u(t)\ubf(t).
\eeq
For the first iteration, $\Wbf$ is a large zero mean matrix independent
of the initial condition $\vbf(0)$.  Also, $\mu_u(0)=0$ and hence
the components of $\zbf_u(0)$ will be
asymptotically Gaussian zero mean variables due to the Central Limit Theorem.
Hence, the vector
$\pbf(0)$ will be distributed as the true vector $\ubf_0$ with Gaussian noise as in
the model \eqref{eq:pqAwgn}.  Therefore, we can take an initial MAP or MMSE
estimate of $\ubf_0$ from the scalar estimation functions in \eqref{eq:GuvMAP}
or \eqref{eq:GuvMMSE}.

Unfortunately, in subsequent iterations with $t>0$,
$\vbf(t)$ is no longer independent of $\Wbf$ and hence $\Wbf\vbf(t)$ will not,
in general, be a Gaussian zero mean random vector.  However, remarkably,
we will show that with the specific choice of $\mu_u(t)$ in Algorithm~\ref{algo:iterMax},
the addition of the term $\mu_u(t)\ubf(t)$ in
\eqref{eq:gamzu} ``debiases" the noise term so that $\zbf_u(t)$ is asymptotically
zero mean Gaussian.  Thus, $\pbf(t)$ continues to be
distributed as in \eqref{eq:pqAwgn} and we can construct MAP or MMSE estimates
of $\ubf_0$ from the vector $\pbf(t)$.  For example, using the MMSE estimation function
\eqref{eq:GuvMMSE} with appropriate selection of the parameters $\lambda_u(t)$,
we obtain that the estimate for $\ubf(\tp1)$ is given by the MMSE estimate
\[
    \ubf(\tp1) = \Exp(\ubf_0|\pbf(t)), \quad \pbf(t)=\gamma_u(t)\ubf_0 + \zbf_u(t).
\]
For the MAP estimation problem,
the MAP selection function \eqref{eq:GuvMAP} computes the MAP estimate for $\ubf_0$.

Similarly, for $\qbf(t)$, we see that
\begin{align}
    \qbf(t) &= \frac{1}{m}\Abf^T\ubf(\tp1)+\mu_v(t)\vbf(t) \nonumber \\
    &= \frac{\ubf_0^T\ubf(\tp1)}{m}\ubf_0
        +  \frac{1}{\sqrt{m}}\Wbf^T\ubf(\tp1) + \mu_v(t)\vbf(t)
        = \gamma_v(r)\ubf_0 + \zbf_v(t), \label{eq:qintuit}
\end{align}
for $\gamma_v(t)=\ubf_0^T\ubf(\tp1)/m$ and
$\zbf_v(t)= (1/\sqrt{m})\Wbf^T\ubf(\tp1) + \mu_v(t)\vbf(t)$.
Then, $\vbf(\tp1)$ is taken as the MMSE or MAP estimate of $\vbf_0$ from
the vector $\qbf(t)$.

Thus, we see that the algorithm attempts to produce a sequence of unbiased
estimates $\pbf(t)$ and $\qbf(t)$ for scaled versions of $\ubf_0$ and $\vbf_0$.
Then, it constructs the MAP or MMSE estimates $\ubf(t)$ and $\vbf(\tp1)$ from
these unbiased estimates.

\subsection{State Evolution Analysis}
The above ``derivation" of the algorithm suggests a possible method
to analyze the IterFac algorithm.  Consider a sequence of problems
indexed by $n$ and suppose that the dimension $m=m(n)$ grows linearly with $n$
in that
\beq \label{eq:betaDef}
    \lim_{n \arr \infty} n/m(n) = \beta
\eeq
for some $\beta > 0$.
For each $n$ and iteration number $t$, define the sets
\beq \label{eq:thetauvi}
    \theta_{u}(t) = \bigl\{ (u_{0i},u_i(t)),
        i=1,\ldots,m \bigr\},  \quad
    \theta_{v}(t) = \bigl\{ (v_{0j},v_j(t)),
        j =1,\ldots,n \bigr\},
\eeq
which are the set of the components of the true vectors $u_{0i}$ and $v_{0j}$
and their corresponding estimates $u_i(t)$ and $v_j(t)$.
To characterize the quality of the estimates, we would like to describe the
distributions of the pairs in $\theta_u(t)$ and $\theta_v(t)$.

Our formal analysis below
will provide an exact characterization of these distributions.
Specifically, we will show that the sets have empirical limits of the
form
\beq  \label{eq:thetauvLim}
    \lim_{n \arr \infty} \theta_{u}(t) \stackrel{d}{=} (U_0,U(t)), \quad
    \lim_{n \arr \infty} \theta_{v}(t) \stackrel{d}{=} (V_0,V(t)),
\eeq
for some random variable pairs $(U_0,U(t)$ and $(V_0,V(t))$.
The precise definition of empirical limits is given in Appendix~\ref{sec:conv},
but loosely, it means that the empirical distribution of the components in $\theta_u(t)$
and $\theta_v(t)$ converge to certain random variable pairs.

In addition, we can inductively derive what the distributions are
for the limits in \eqref{eq:thetauvLim}.
To make matters simple, suppose that the parameters
$\lambda_u(t)$ and $\lambda_v(t)$ are not
selected adaptively but instead given by a fixed sequences $\lambdabar_u(t)$ and
$\lambdabar_v(t)$.  Also, given random variables in the limits of \eqref{eq:thetauvLim},
define the deterministic constants
\beq \label{eq:alphaSE}
    \alphabar_{u0}(t) = \Exp\left[ U(t)^2 \right], \quad 
    \alphabar_{u1}(t) = \Exp\left[ U_0U(t) \right], \quad 
    \alphabar_{v0}(t) = \Exp\left[ V(t)^2 \right], \quad 
    \alphabar_{v1}(t) = \Exp\left[ V_0V(t) \right].
\eeq

Now suppose that the second limit in \eqref{eq:thetauvLim} holds for some $t$.
Then, $\gamma_u(t)$ in \eqref{eq:gamzu} would have the following limit:
\[
    \lim_{n \arr \infty} \gamma_u(t) = \lim_{n\arr \infty} \frac{n}{m} \frac{1}{n}
        \sum_{j=1}^n v_{0j}v_j(t) = \beta \alphabar_{v1}(t).
\]
Also, if we ignore the debiasing term with $\mu_u(t)$ and assume
that $\Wbf$ is independent of $\vbf(t)$, the variance of the components of
$\zbf_u(t)$ in \eqref{eq:gamzu} would be
\[
    \var(z_{ui}(t)) = \frac{1}{m}\sum_{j=1}^n \var(W_{ij}v_j(t))
        = \beta \tau_w \alphabar_{v0}(t).
\]
Thus, we would expect a typical component of $p_i(t)$ in \eqref{eq:pintuit}
to be distributed as
\[
    P(t) = \beta \alphabar_{v1}(t)U_0+Z_u(t), \quad 
     Z_u(t)  \sim  {\mathcal N}(0,\beta\tau_w\alphabar_{v0}(t)).
\]
Due to the separability assumption \eqref{eq:Guvsep}, each
component $u_i(\tp1) = G_u(p_i(t),\lambdabar_u(t))$.  So, we would expect
the components to follow the distribution
\beq \label{eq:USE}
   U(\tp1) = G_u(P(t),\lambdabar_u(t)), \quad                 
      P(t) = \beta \alphabar_{v1}(t)U_0+Z_u(t), \quad 
 Z_u(t)  \sim  {\mathcal N}(0,\beta\tau_w\alphabar_{v0}(t)). 
\eeq
This provides an exact description of the expected joint density of $(U_0,U(\tp1))$.
From this density we can then compute $\alphabar_{u0}(\tp1)$
and $\alphabar_{u1}(\tp1)$ in \eqref{eq:alphaSE}.

A similar calculation, again assuming the ``debiasing" and Gaussianity
assumptions are valid, shows that the limiting empirical distribution
of $\theta_v(\tp1)$ in \eqref{eq:thetauvLim} should follow
\beq \label{eq:VSE}
    V(\tp1) =  G_v( Q(t),\lambdabar_v(t)), \quad 
       Q(t) =  \alphabar_{u1}(\tp1)V_0+Z_v(t), \quad 
  Z_v(t) \sim  {\mathcal N}(0,\tau_w\alphabar_{u0}(\tp1) ). 
\eeq
This provides the joint density $(V_0,V(\tp1))$ from which
we can compute $\alphabar_{v0}(\tp1)$
and $\alphabar_{v1}(\tp1)$ in \eqref{eq:alphaSE}.
Thus, we have provided a simple method to recursively compute
the joint densities of the limits in \eqref{eq:thetauvLim} and their
second-order statistics \eqref{eq:alphaSE}.

\section{Asymptotic Analysis under Gaussian Noise}
\label{sec:asymAnal}

\subsection{Main Results}

In the above intuitive analysis, we did not formally describe the sense
of convergence nor offer any formal justification for the Gaussianity
assumptions.  In addition, we assumed that the parameters $\lambda_u(t)$
and $\lambda_v(t)$ were fixed sequences.  In reality, we will need them to
be data adaptive.  We will make the above arguments rigorous under the following
assumptions.  Note that
although we will be interested in MAP or MMSE estimation functions,
our analysis will apply to arbitrary factor selection functions $G_u(\cdot)$
and $G_v(\cdot)$.

\begin{assumption} \label{as:rankOne}
Consider a sequence of random realizations of the
estimation problem in Section~\ref{sec:intro}
indexed by the dimension $n$.
The matrix $\Abf$ and the parameters in Algorithm~\ref{algo:iterMax}
satisfy the following:
\begin{enumerate}[(a)]
\item For each $n$, the output dimension $m=m(n)$
is deterministic and scales linearly with $n$ as \eqref{eq:betaDef}.

\item
The matrix $\Abf$ has the form \eqref{eq:ArankOne}
where $\ubf_0 \in \R^m$ and $\vbf_0 \in \R^n$ represent
``true" left and right factors of a rank one term,
and  $\Wbf \in \R^{m \x n}$
is an i.i.d.\ Gaussian matrix with
components $W_{ij} \sim {\mathcal N}(0,\tau_w)$ for some $\tau_w > 0$.

\item  The factor selection functions $G_u(\pbf,\lambda_u)$
and $G_v(\qbf,\lambda_v)$
in lines \ref{line:Guopt} and \ref{line:Gvopt}
are componentwise separable in that for all component indices $i$ and $j$,
\beq\label{eq:GuvSep}
    G_u(\pbf,\lambda_u)_i = G_u(p_i,\lambda_u), \quad
    G_v(\qbf,\lambda_v)_j = G_v(q_j,\lambda_v),
\eeq
for some scalar functions $G_u(p,\lambda_u)$ and $G_v(q,\lambda_v)$.
The scalar functions
must be differentiable in $p$ and $q$.  Moreover, for every $t$,
the functions $G_u(p,\lambda_u)$ and
$\partial G_u(p,\lambda_u)/\partial p$
must be Lipschitz continuous in $p$ with a Lipschitz constant
that is continuous in $\lambda_u$, and continuous in $\lambda_u$
uniformly over $p$.  Similarly, for every $t$,
the functions $G_v(q,\lambda_v)$ and $\partial G_v(q,\lambda_v)/\partial q$
must be Lipschitz continuous in $q$ with a Lipschitz constant
that is continuous in $\lambda_v$, and continuous in $\lambda_v$
uniformly over $q$.

\item  The parameters $\lambda_u(t)$ and $\lambda_v(t)$
are computed via
\beq \label{eq:lamRankOne}
    \lambda_u(t) = \frac{1}{n} \sum_{j=1}^n
        \phi_{\lambda v}(t,v_{0j},v_j(t)), \quad 
    \lambda_v(t) = \frac{1}{m} \sum_{i=1}^m
        \phi_{\lambda u}(t,u_{0j},u_j(\tp1)) 
\eeq
for (possibly vector-valued) functions $\phi_{\lambda u}(\cdot)$
and $\phi_{\lambda v}(\cdot)$ that are pseudo-Lipschitz continuous
of order $p=2$.

\item
For $t=0$, the sets \eqref{eq:thetauvLim}
empirically converge with bounded moments of order $2$
to the limits
\beq \label{eq:thetauvLimInit}
    \lim_{n \arr \infty} \theta_u(0) \stackrel{d}{=} (U_0,U(0)), \quad
    \lim_{n \arr \infty} \theta_v(0) \stackrel{d}{=} (V_0,V(0)),
\eeq
for some random variable pairs $(U_0,U(0)$ and $(V_0,V(0))$.
See Appendix~\ref{sec:conv} for a precise definition of
the empirical convergence used here.
\end{enumerate}
\end{assumption}

The assumptions need some explanations.
Assumptions~\ref{as:rankOne}(a) and (b) simply state that we are
considering an asymptotic analysis for certain large matrices $\Abf$
consisting of a random rank one matrix plus Gaussian noise.
The analysis of Algorithm~\ref{algo:iterMax}
for higher ranks is still not known, but we provide some possible
ideas later.
Assumption~\ref{as:rankOne}(c) is a mild condition on the factor
selection functions.  In particular, the separability assumption
holds for the MAP or MMSE functions
\eqref{eq:GuvMAP} and \eqref{eq:GuvMMSE} under separable priors.

Assumption (d) allows for the parameters $\lambda_u(t)$ and $\lambda_v(t)$
in the factor selection functions to be data dependent, provided that
they can each be determined via empirical averages of some function
of the most recent data.  Assumption (e) is the initial induction hypothesis.

Under these assumptions, we can recursively
define the sequence of random variables $(U_0,U(t))$
and $(V_0,V(t))$ as described above.  For the parameters $\lambdabar_u(t)$
and $\lambdabar_v(t)$, define them from the expectations
\beq \label{eq:lambdaSE}
    \lambdabar_{u}(t) = \Exp\left[ \phi_{\lambda u}(V_0,V(t)) \right], \quad
    \lambdabar_{v}(t) = \Exp\left[ \phi_{\lambda v}(U_0,U(\tp1)) \right].
\eeq
These are the limiting values we would expect given the adaptation rules
\eqref{eq:lamRankOne}.

\begin{theorem} \label{thm:rankOneLim}  Under Assumption~\ref{as:rankOne},
the sets $\theta_u(t)$ and $\theta_v(t)$
in \eqref{eq:thetauvi} converge empirically with bounded moments of
order $p=2$ with the limits in \eqref{eq:thetauvLim}.
\end{theorem}
\begin{proof}  See Appendix~\ref{sec:rankOneLimPf}.
\end{proof}

\subsection{Scalar Equivalent Model}
The main contribution of Theorem~\ref{thm:rankOneLim}
is that it provides a simple \emph{scalar equivalent model} for
the asymptotic behavior of the algorithm.
The sets $\theta_u(t) = \{(u_{0i},u_i(t))\}$
and $\theta_v(t) = \{ (v_{0j},v_j(t)) \}$ in \eqref{eq:thetauvi}
are the components of true vectors $\ubf_0$ and $\vbf_0$
and their estimates $\ubf(t)$ and $\vbf(t)$.
The theorem shows that empirical distribution of these
components are asymptotically equivalent to simple random variable pairs
$(U_0,U(t))$ and $(V_0,V(t))$ given by \eqref{eq:USE} and \eqref{eq:VSE}.
In this scalar system, the variable $U(\tp1)$ is the output of the factor selection
function $G_u(\cdot)$ applied to a scaled and Gaussian noise-corrupted version
of the true variable $U_0$.  Similarly,
$V(\tp1)$ is the output of the factor selection
function $G_v(\cdot)$ applied to a scaled and Gaussian noise-corrupted version
of the true variable $V_0$.
Following \citep{GuoBS:09-Allerton}, we can thus call the result
a \emph{single-letter} characterization of the algorithm.

From this single-letter characterization, one can exactly compute a
large class of performance metrics of the algorithm.  Specifically,
the empirical convergence of $\theta_u(t)$
shows that for any pseudo-Lipschitz function $\phi(u_0,u)$
of order $p$, the following limit exists almost surely:
\beq \label{eq:phiuLim}
    \lim_{n \arr \infty} \frac{1}{m} \sum_{i=1}^m \phi(u_{0i},u_i(t)) =
        \Exp\left[ \phi(U_0,U(t)) \right],
\eeq
where the expectation on the right-hand side is over the variables
$(U_0,U(t))$ with $U_0$ identical to the variable in the limit
in \eqref{eq:thetauvLimInit} and $U(t)$ given by \eqref{eq:USE}.
This expectation can thus be explicitly evaluated by a simple
two-dimensional integral and consequently any component-separable
performance metric based on a suitably continuous
loss function $\phi(u_0,u)$ can be exactly computed.

For example, if we take $\phi(u_0,u) = (u-u_0)^2$
we can compute the asymptotic mean squared error of the estimate,
\beqa
    \lim_{n \arr \infty} \frac{1}{m} \|\ubf(t)-\ubf_0\|^2 =
    \lim_{n \arr \infty} \frac{1}{m} \sum_{i=1}^m (u_i(t)-u_{0i})^2 \nonumber =
    \Exp\left[ (U_0-U(t))^2 \right]. \label{eq:mseuLim}
\eeqa
Also, for each $t$, define the empirical second-order statistics
\beq \label{eq:alphauv}
    \alpha_{u0}(t) = \frac{1}{m}\|\ubf(t)\|^2,  ~
    \alpha_{u1}(t) = \frac{1}{m}\ubf(t)^T\ubf_0, ~
    \alpha_{v0}(t) = \frac{1}{n}\|\vbf(t)\|^2,  ~
    \alpha_{v1}(t) = \frac{1}{n}\vbf(t)^T\vbf_0.
\eeq
Since $\|\ubf(t)\|^2 = \sum_i u_i(t)^2$, it follows that
$\alpha_{u0}(t) \arr \Exp(U(t)^2)$ almost surely as $n \arr \infty$.
In this way, we obtain that the following limits hold almost surely:
\beq \label{eq:alphauvLim}
    \lim_{n \arr \infty} \alpha_{u0}(t) = \alphabar_{u0}(t), ~
    \lim_{n \arr \infty} \alpha_{u1}(t) = \alphabar_{u1}(t),  ~
    \lim_{n \arr \infty} \alpha_{v0}(t) = \alphabar_{v0}(t),  ~
    \lim_{n \arr \infty} \alpha_{v1}(t) = \alphabar_{v1}(t).  
\eeq
We will also use definitions
\beq \label{eq:tauuv}
    \tau_u := \Exp[U_0^2], \quad \tau_v := \Exp[V_0^2].
\eeq
From the second order statistics, we can compute
the asymptotic correlation coefficient between $\ubf_0$
and its estimate $\ubf$  given by
\beqa
    \rho_u(t) &:=& \lim_{n \arr \infty} \frac{|\ubf(t)^T\ubf_0|^2}
    {\|\ubf(t)\|^2\|\ubf_0\|^2}
    = \lim_{n \arr \infty} \frac{|(\ubf(t)^T\ubf_0)/m|^2}
    {(\|\ubf(t)\|^2/m)(\|\ubf_0\|^2/m)} \nonumber \\
    &=&  \frac{ \bigl[\Exp(U(t)U_0)\bigr]^2 }{
    \Exp U(t)^2 \Exp U_0^2} = \frac{\alphabar_{u1}^2(t)}{\alphabar_{u0}(t)\tau_u}.
     \hspace{0.5in} \label{eq:corruLim}
\eeqa
Similarly, the asymptotic correlation coefficient between $\vbf_0$ and
$\vbf$ has a simple expression
\beq \label{eq:corrvLim}
    \rho_v(t) := \lim_{n \arr \infty} \frac{|\vbf(t)^T\vbf_0|^2}
    {\|\vbf(t)\|^2\|\vbf_0\|^2}
     = \frac{\alphabar_{v1}^2(t)}{\alphabar_{v0}(t)\tau_v}.
\eeq
The correlation coefficient is useful, since we know that, without additional constraints,
the terms $\ubf_0$ and $\vbf_0$ can only be estimated up to a scalar.  The correlation
coefficient is scale invariant.

\section{Examples} \label{sec:examples}

The SE analysis can be used to exactly predict the asymptotic
behavior of the IterFac algorithm for any smooth scalar estimation
functions, including the MAP or MMSE functions \eqref{eq:GuvMAP} and \eqref{eq:GuvMMSE}.
There are, however, two cases, where the SE equations have particularly
simple and interesting solutions:  linear estimation functions and MMSE functions.

\subsection{Linear Selection Functions} \label{sec:linSel}

Suppose we use linear selection functions of the form
\beq \label{eq:Guvlin}
    G_u(\pbf,\lambda_u) =\lambda_u\pbf,\quad
    G_v(\qbf,\lambda_v)= \lambda_v\qbf,
\eeq
where the parameters $\lambda_u$ and $\lambda_v$
allow for normalization or other scalings of the outputs.
Linear selection functions of the form \eqref{eq:Guvlin}
arise when one selects
$G_u(\cdot)$ and $G_v(\cdot)$ from the MAP or MMSE functions \eqref{eq:GuvMAP}
or \eqref{eq:GuvMMSE} with Gaussian priors.

With Gaussian priors, the correct solution to the MAP estimate \eqref{eq:uvMAP}
is for $(\ubfhat,\vbfhat)$ to be the (appropriately scaled)
left and right maximal singular vectors of $\Abf$. We will
thus call the estimates $(\ubf(t),\vbf(t))$
of Algorithm~\ref{algo:iterMax} and linear selection functions \eqref{eq:Guvlin}
the \emph{estimated maximal singular vectors}.

\medskip
\begin{theorem} \label{thm:linConv}  Consider the state evolution equation
\eqref{eq:alphaSE},
\eqref{eq:USE}, and \eqref{eq:VSE} with the
linear selection functions \eqref{eq:Guvlin}.  Then:
\begin{itemize}
\item[(a)] The asymptotic
correlation coefficients \eqref{eq:corruLim} and \eqref{eq:corrvLim}
satisfy the following recursive rules:
\beq  \label{eq:corrLinSE}
    \rho_u(\tp1) = \frac{\beta \tau_u\tau_v \rho_v(t)}{
        \beta \tau_u\tau_v \rho_v(t) + \tau_w}, \quad
    \rho_v(t) = \frac{\tau_u\tau_v \rho_u(t)}{\tau_u\tau_v \rho_u(t) + \tau_w}.
\eeq
\item[(b)]  For any positive initial condition, $\rho_v(0) > 0$,
the asymptotic correlation coefficients converge to the limits
\beq \label{eq:corrLinLimit}
    \lim_{t \arr \infty} \rho_u(t) = \rho_u^* :=
    \frac{\left[ \beta\tau_u^2\tau_v^2 - \tau_w^2 \right]_+}{
    \tau_u\tau_v(\beta\tau_u\tau_v + \tau_w)}, \quad
   \lim_{t \arr \infty} \rho_v(t) = \rho_v^* :=
    \frac{\left[\beta\tau_u^2\tau_v^2 - \tau_w^2\right]_+}{
    \beta \tau_u\tau_v(\tau_u\tau_v + \tau_w)}, 
\eeq
where $[x]_+= \max\{0,x\}$.
\end{itemize}
\end{theorem}
\begin{proof}  See Appendix~\ref{sec:linConvPf}.
\end{proof}

The theorem provides a set of recursive equations for
the asymptotic correlation coefficients $\rho_u(t)$ and $\rho_v(t)$
along with simple expressions for the limiting values as $t \arr \infty$.
We thus obtain exactly how correlated the
estimated maximal singular vectors of a matrix $\Abf$ of the form
\eqref{eq:ArankOne} are to the rank one factors $(\ubf_0,\vbf_0)$.
The proof of the theorem also provides expressions for the second-order
statistics in \eqref{eq:alphaSE} to be used in the scalar equivalent model.

The fixed point expressions \eqref{eq:corrLinLimit}
agree with the more general results in \citep{KesMonOh:10}
that derive the correlations for ranks greater than one and low-rank
recovery with missing entries.
Similar results can also be found in \citep{CapDonFer:09}.
An interesting consequence of the expressions in \eqref{eq:corrLinLimit}
is that unless
\beq \label{eq:snrMinLin}
    \sqrt{\beta}\tau_u\tau_v > \tau_w,
\eeq
the asymptotic correlation coefficients are exactly zero.
The ratio $\tau_u\tau_v/\tau_w$ can be interpreted as a scaled SNR.

\subsection{Minimum Mean-Squared Error Estimation} \label{sec:mmse}

Next suppose we use the MMSE selection functions \eqref{eq:GuvMMSE}.
Using the scalar equivalent models \eqref{eq:USE} and \eqref{eq:VSE},
we take the scaling factor and noise parameters as
\begin{align}  \label{eq:mmseParam}
\begin{split}
    &\gamma_u(t) = \beta \alphabar_{v1}(t), \quad \nu_u(t)=\beta \tau_w \alphabar_{v0}(t) \\
    &\gamma_v(t) = \alphabar_{u1}(\tp1), \quad \nu_v(t)= \tau_w \alphabar_{u0}(\tp1).
\end{split}
\end{align}
Observe that these parameters can be computed from the SE equations and hence can determined
offline and are thus not data dependent.
We can use the initial condition $v_j(0) = \Exp[V_0]$ for all $j$,
so that the initial
variable in \eqref{eq:thetauvLimInit} is $V(0) = \Exp[V_0]$.
To analyze the algorithms define
\beq \label{eq:mmseuv}
    {\mathcal E}_u(\eta_u) = \var(U_0 \MID Y=\sqrt{\eta_u}U_0 + D),
     \quad
    {\mathcal E}_v(\eta_v) = \var(V_0 \MID Y=\sqrt{\eta_v}V_0 + D),
\eeq
where $D \sim \mathcal{N}(0,1)$ is independent of $U_0$ and $V_0$.
That is, ${\mathcal E}_u(\eta_u)$ and ${\mathcal E}_v(\eta_v)$
are the mean-squared errors of estimating $U_0$ and $V_0$ from
observations $Y$ with SNRs of $\eta_u$ and $\eta_v$.
The functions ${\mathcal E}_u(\cdot)$ and ${\mathcal E}_v(\cdot)$
arise in a range of estimation problems and the analytic and functional
properties of these functions can be found in \citep{GuoWuSV:11,WuVerdu:12}.

\begin{theorem} \label{thm:SEmmse}
Consider the solutions to the SE equations \eqref{eq:USE}, \eqref{eq:VSE},
and \eqref{eq:alphaSE}
under the MMSE selection functions \eqref{eq:GuvMMSE} with parameters \eqref{eq:mmseParam}
and initial condition
$V(0) = \Exp[V_0]$.  Then:
\begin{itemize}
\item[(a)] For all $t$, the asymptotic correlation coefficients
\eqref{eq:corruLim} and \eqref{eq:corrvLim}  satisfy the recursive
relationships
\beq \label{eq:corrSEmmse}
    \rho_u(\tp1) = 1 - \frac{1}{\tau_u}{\mathcal E}_u( \beta\tau_v\rho_v(t)/\tau_w), \quad 
    \rho_v(t) = 1 - \frac{1}{\tau_v}{\mathcal E}_v( \tau_u\rho_u(t)/\tau_w),
\eeq
with the initial condition $\rho_v(0)=(\Exp V_0)^2/\tau_v$.

\item[(b)]  If, in addition, $\Ecal_u(\eta_u)$
and $\Ecal_v(\eta_v)$ are continuous, then
for any positive initial condition, $\rho_v(0)>0$,
as $t \arr \infty$, the asymptotic correlation coefficients
$(\rho_u(t),\rho_v(t))$ increase monotonically to
fixed points $(\rho_u^*, \rho_u^*)$ of \eqref{eq:corrSEmmse}
with $\rho_v^* > 0$.
\end{itemize}
\end{theorem}
\begin{proof} See Appendix \ref{sec:SEmmsePf}. \end{proof}

Again, we see that we can obtain simple, explicit recursions
for the asymptotic correlations.  Moreover, the asymptotic correlations
provably converge to fixed points of the SE equations.  The
proof of the theorem also provides expressions for the second-order
statistics in \eqref{eq:alphaSE} to be used in the scalar equivalent model.

\subsection{Zero Initial Conditions} \label{sec:zero}
The limiting condition in part (b) of Theorem \ref{thm:SEmmse} requires that
$\rho_v(0) > 0$, which occurs when $\Exp[V_0] \neq 0$.
Suppose, on the other hand, that  $\Exp[U_0] = \Exp[V_0] = 0$.
Then, the initial condition will
be $V(0)=\Exp[V_0] = 0$. Under this initial condition,  a simple set of calculations
show that the SE equations \eqref{eq:corrSEmmse} will generate a sequence
 with $\rho_v(t) = \rho_u(t)=0$
for all $t$.  Thus, the IterFac algorithm will produce no useful estimates.

Of course, with zero mean random variables, a more sensible initial condition is
to take $\vbf(0)$ to be some non-zero random vector, as is commonly done
in power algorithm recursions for computing maximal singular vectors.
To understand the
behavior of the algorithm under this random initial condition, let
\beq \label{eq:rhovn}
    \rho_v(t,n) := \frac{|\vbf(t)^T\vbf_0|^2}{\|\vbf_0\|^2 \|\vbf(t)\|^2},
\eeq
where we have explicitly denoted the dependence on the problem dimension $n$.
From \eqref{eq:corrvLim}, we have that $\lim_{n \arr\infty} \rho_v(t,n) = \rho_v(t)$
for all $t$.  Also, with a random initial condition $\vbf(0)$ independent of
$\vbf_0$, it can be checked that $\rho_v(0,n) = O(1/n)$ so that
\[
    \rho_v(0) = \lim_{n \arr\infty} \rho_v(0,n) = 0.
\]
Hence, from the SE equations \eqref{eq:corrSEmmse}, $\rho_v(t) = \rho_u(t)=0$
for all $t$.  That is,
\beq \label{eq:corrLimZero}
    \lim_{t \arr \infty} \lim_{n \arr\infty} \rho(t,n) = 0.
\eeq
This limit suggests that, even with random initial condition, the IterFac algorithm
will not produce a useful estimate.

However, it is still possible that the limit in the opposite order
of \eqref{eq:corrLimZero} may be non-zero:
\beq \label{eq:rhontLim}
    \lim_{n \arr \infty} \lim_{t \arr\infty} \rho(t,n) > 0.
\eeq
That is, for each $n$, it may be possible to obtain a
non-zero correlation, but the number of iterations for
convergence increases with $n$ since the algorithm starts from a decreasingly small
initial correlation.  Unfortunately, our SE analysis cannot make predictions
on limits in the order of \eqref{eq:rhontLim}.

We can however analyze the following limit:

\begin{lemma} \label{lem:mmseZero}
Consider the MMSE SE equations \eqref{eq:corrSEmmse} with random variables
$U_0$ and $V_0$ such that $\Exp[V_0] = \Exp[U_0] = 0$.
For each $\epsilon > 0$,
let $\rho_v^\epsilon(t)$ be the solution to the SE equations with
an initial condition $\rho_v(0)= \epsilon$.
Then,
\begin{itemize}
\item[(a)]  If $\sqrt{\beta} \tau_u \tau_v > \tau_w$,
\beq \label{eq:corrvLimEpsPos}
    \lim_{\epsilon \arr 0} \lim_{t \arr \infty} \rho^\epsilon_v(t) > 0.
\eeq
\item[(b)]  Conversely, if $\sqrt{\beta} \tau_u \tau_v < \tau_w$,
\beq \label{eq:corrvLimEpsNeg}
    \lim_{\epsilon \arr 0} \lim_{t \arr \infty} \rho^\epsilon_v(t) = 0.
\eeq
\end{itemize}
\end{lemma}
\begin{proof} See Appendix \ref{sec:mmseZeroPf}. \end{proof}

The result of the lemma is somewhat disappointing.
The lemma shows that $\sqrt{\beta} \tau_u \tau_v > \tau_w$ is essentially
necessary and sufficient for the IterFac algorithm with MMSE estimates
to be able to overcome arbitrarily small initial conditions and obtain
an estimate with a non-zero correlation to the true vector.
Unfortunately, this is the identical to the condition
\eqref{eq:snrMinLin} for the linear estimator to obtain a non-zero correlation.
Thus, the IterFac algorithm with MMSE estimates performs no better
than simple linear estimation in the initial iterations when the priors have zero means.
Since linear estimation
is equivalent to finding maximal singular vectors without any particular constraints,
we could interpret Lemma \ref{lem:mmseZero}
as saying that the IterFac algorithm under MMSE estimation cannot exploit
structure in the components in the initial iterations.
As a result, in low SNRs it may be necessary to use other algorithms
as an initial condition for IterFac -- such procedures, however, require further study.
\section{Numerical Simulation} \label{sec:sim}
To validate the SE analysis, we consider a simple case where
the left factor $\ubf_0 \in \R^m$ is i.i.d.\ Gaussian, zero mean and
$\vbf_0 \in \R^n$ has Bernoulli-Exponential components:
\beq \label{eq:vBernExp}
    v_{0j} \sim \left\{ \begin{array}{ll}
        0 & \mbox{with prob } 1-\lambda, \\
        \mbox{Exp}(1) & \mbox{with prob } \lambda,
        \end{array} \right.
\eeq
which provides a simple model for a sparse, positive vector.  The parameter
$\lambda$ is the fraction of nonzero components and is set in
this simulation to $\lambda=0.1$.
Note that these components have a non-zero mean so the difficulties of
Section \ref{sec:zero} are avoided.
The dimensions are $(m,n)=(1000,500)$, and the noise level $\tau_w$ is set
according to the scaled SNR defined as
\beq \label{eq:snrDef}
    \SNR = 10\log_{10}(\tau_u\tau_v/\tau_w).
\eeq
Estimating the vector $\vbf_0$ in this set-up
is related to finding sparse principal vectors of the matrix $\Abf^T\Abf$
for which there are large number of excellent methods including
\citep{CadJoll:95,JollTreUdd:03,ZhaZhaS:02,ZouHasTib:06,AspEGJL:07,ShenaHuang:08}
to name a few.  These algorithms include methods based on thresholding,
$\ell_1$-regularization and semidefinite programming.  A comparison
of the IterFac against these methods would be an interesting
avenue of future research.  Here, we simply wish to verify the SE predictions
of the IterFac method.

\begin{figure}
\centering
\includegraphics[width=3.5in]{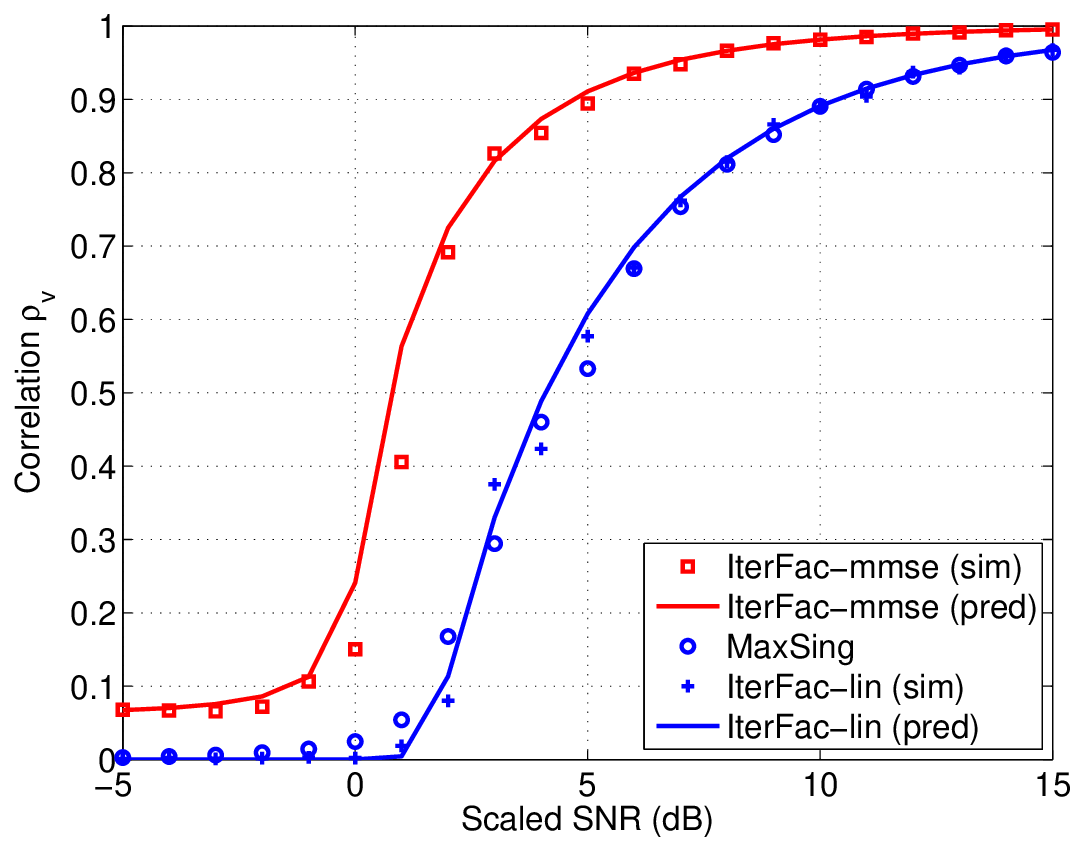}

\caption{Simulation of the IterFac algorithm for both the linear selection
functions in Section~\ref{sec:linSel} (labeled iter-lin) and MMSE selection
functions in Section \ref{sec:mmse} (labeled iter-mmse).
Plotted are the correlation values after 10 iterations.
The simulated values are compared against the SE predictions.
Also shown is the simple estimate from the maximal singular vectors of $\Abf$.
\label{fig:rankOneSim} }
\end{figure}

\begin{figure}
\begin{center}
\includegraphics[width=3.5in]{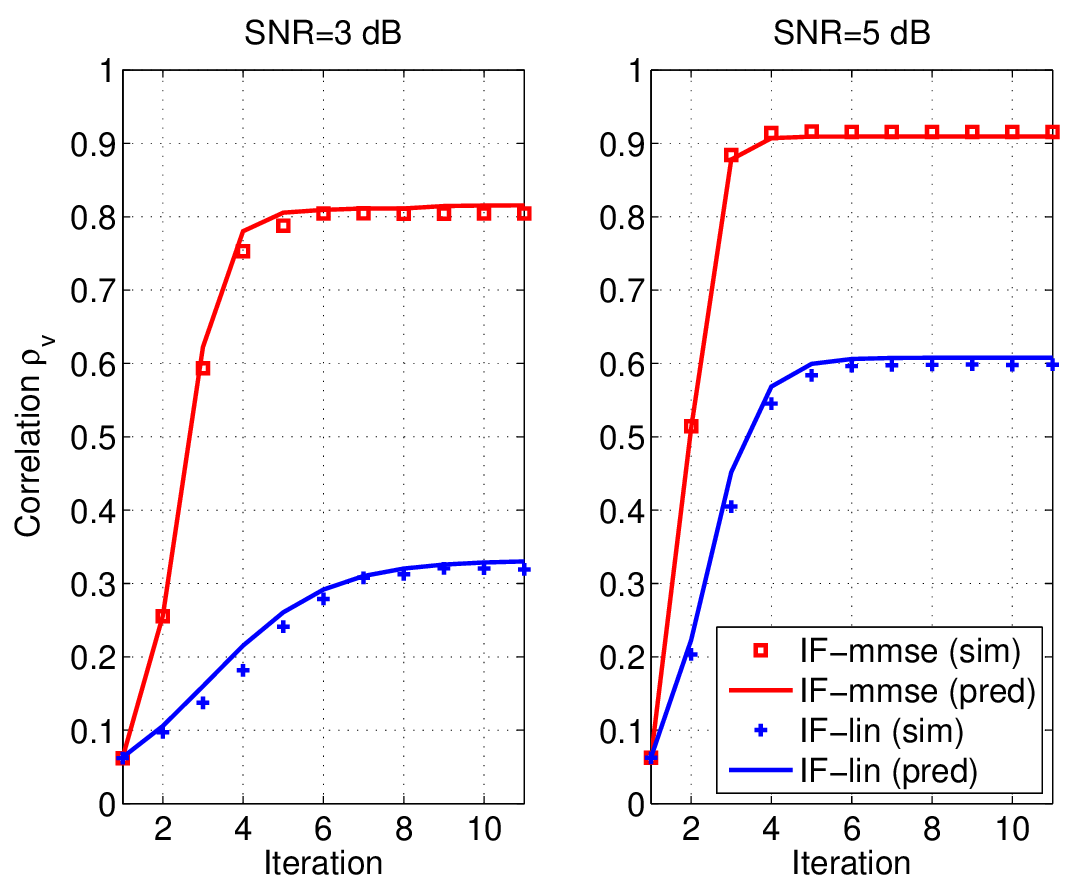}
\end{center}

\caption{Per iteration performance of the
IterFac algorithm for both the linear selection
functions in Section~\ref{sec:linSel} (IF-lin) and MMSE selection functions in Section
\ref{sec:mmse} (IF-mmse).  The simulated values are compared against the SE predictions.
\label{fig:rankOneIter} }
\end{figure}

The results of the simulation are shown in Fig.~\ref{fig:rankOneSim},
which shows the simulated and SE-predicted performance of the IterFac
algorithm with both the linear and MMSE selection functions for the priors
on $\ubf$ and $\vbf$.
The algorithm is run for $t=10$ iterations and the plot shows the median of the
final correlation coefficient $\rho_v(t)$
over 50 Monte Carlo trials at each value of SNR\@.
It can be seen that the performance of the IterFac algorithm
for both the linear and MMSE estimates are in excellent agreement with the
SE predictions.
The correlation coefficient of the linear estimator also matches
the correlation
of the estimates produced from the maximal singular vectors of $\Abf$.
This is not surprising since, with linear selection functions, the IterFac
algorithm is essentially an iterative method to find the maximal singular vectors.
The figure also shows the benefit of exploiting the prior on $\vbf_0$,
which is evident from the superior performance of the MMSE estimate over the linear
reconstruction.

Fig.\ \ref{fig:rankOneIter} shows the correlation coefficient as a function
of the iteration number for the MMSE and linear methods for two values of the SNR.
Again, we see that the SE predictions are closely matched to the median
simulated values.  In addition, we see that we get good convergence within 4 to 8 iterations. Based on the SE predictions, this number will not scale with dimension
and hence the simulation suggests that only a small number of iterations will be
necessary for even very large problems.
All code for the simulation can be found in the public GAMP sourceforge
project \citep{GAMP-code}.

\section{Limitations and Extensions} \label{sec:limits}

There are several potential lines for future work --
some of which have already been explored in other works made since the
original publication of this paper in \citep{RanganF:12-ISIT}.

\paragraph{Extensions to higher rank}
The algorithm presented in this paper considers only
rank one matrices.
The works
\citep{sakata2013sample,krzakala2013phase,kabashima2014phase,parker2014bilinear}
have proposed AMP-type algorithms for more general classes
of matrix factorization problems and provided analyses of these methods
based on replica methods and other ideas from statistical physics.

\paragraph{Unknown priors}
The MMSE estimator in Section \ref{sec:mmse} requires exact knowledge
of the priors on $U_0$ and $V_0$ as well as the Gaussian noise level $\tau_w$.
In many problems in statistics, these are not known.  There are two possible
solutions that may be investigated in the future.
One method is to parameterize the distributions of $U_0$ and $V_0$
and estimate these parameters in the MMSE selection functions
\eqref{eq:GuvMMSE}
-- potentially through an EM type procedure as in \citep{DempLR:77}.
This EM type approach with hidden hyperparameters has been recently
successfully used in a related approximate message passing method in
\citep{VillaSch:11}.  The analysis of the such parameter learning could
possibly be accounted for through the adaptation parameters
$\lambda_u(t)$ and $\lambda_v(t)$.
A second approach is to assume that the distributions of
$U_0$ and $V_0$ belong to a known family of distributions and then find
a min-max solution.  Such a min-max technique was proposed for AMP
recovery of sparse vectors in \citep{DonohoMM:09}.  See also
\citep{DonohoJ:94b}.

\paragraph{Optimality}
While the current paper characterizes the performance of the
IterFac algorithm, it remains open how far that performance is to optimal
estimation such as the joint MMSE estimates of $\ubf_0$ and $\vbf_0$.
AMP methods have been shown to be provably optimal in a wide range
of scenarios \citep{BouCaire:02,GuoW:06,GuoW:07,Rangan:10arXiv,Rangan:11-ISIT,BayatiM:11,javanmard2013state}.
More recently, \citep{deshpande2014information} has proven
that for
sparse binary priors, the IterFac algorithm was provably Bayes optimal
for certain sparsity levels.
A potential line of future work is to see if these results can be extended
to more general settings.

\section*{Conclusions}
We have presented a computationally-efficient method for estimating
rank-one matrices in noise.  The estimation problem is reduced to a sequence of
scalar AWGN estimation problems which can be performed easily for a large class
of priors or regularization functions on the coefficients.  In the case of
Gaussian noise, the asymptotic performance of the algorithm is exactly
characterized by a set of scalar state evolution equations which appear
to match the performance at moderate dimensions well.
Thus, the methodology is computationally simple, general and admits a precise
analysis in certain asymptotic, random settings.
Future work include extensions to higher rank matrices and
handling of the cases where the priors are not known.

\section*{Acknowledgments}
The authors thank Vivek Goyal, Ulugbek Kamilov,
Andrea Montanari and Phil Schniter for detailed comments on an earlier draft.

\appendix
\section{Appendices:  Proofs}

\subsection{Empirical Convergence of Random Variables} \label{sec:conv}
Bayati and Montanari's analysis in \citep{BayatiM:11}
employs certain deterministic
models on the vectors and then proves convergence properties
of related empirical distributions. To apply the same analysis here, we need
to review some of their definitions.  We say a function $\phi:\R^r \arr \R^s$
is \emph{pseudo-Lipschitz} of order $p>1$, if there exists an $L > 0$
such for any $\xbf$, $\ybf \in \R^r$,
\[
    \|\phi(\xbf) - \phi(\ybf)\|
        \leq L(1+\|\xbf\|^{p-1}+\|\ybf\|^{p-1})\|\xbf-\ybf\|.
\]

Now suppose that for each $n=1,2,\ldots$, we have a set of vectors
\[
    \theta(n) = \left\{ \vbf_i(n), i=1,\ldots,\ell(n)\right\},
\]
where the elements are vectors $\vbf_i(n) \in \R^s$, and the size
of the set is given by $\ell(n)$.
We say that the set of components of $\theta(n)$
\emph{empirically converges with bounded moments of order $p$} as $n\arr\infty$
to a random vector $\Vbf$ on $\R^s$ if:  For all
pseudo-Lipschitz continuous functions, $\phi$, of order $p$,
\beq \label{eq:phiConv}
    \lim_{n \arr\infty} \frac{1}{\ell(n)} \sum_{i=1}^{\ell(n)} \phi(\vbf_i(n))
    = \Exp[\phi(\Vbf)] < \infty.
\eeq
When the nature of convergence
is clear, we may write (with some abuse of notation)
\[
      \lim_{n \arr \infty} \theta(n) \stackrel{d}{=} \Vbf.
\]

\subsection{Bayati--Montanari Recursions with Adaptation}

Our main result will need an \emph{adaptive} version of the
recursion theorem of Bayati and Montanari
\citep{BayatiM:11}.
Let $H_u(t,d,u_0,\nu_u)$ and $H_v(t,b,v_0,\nu_v)$ be two functions
defined on arguments $t=0,1,2,\ldots$ and
$d$, $b$, $u_0$ and $v_0 \in \R$ as well as vectors $\nu_u$ and $\nu_v$.
Given a matrix $\Sbf \in \R^{m \x n}$ and vectors
$\ubf_0$ and $\vbf_0$, generate a sequence of
vectors $\bbf(t)$ and $\dbf(t)$ by the iterations
\begin{subequations} \label{eq:genAMPRec}
\beqa
    \bbf(t) &=& \Sbf \vbf(t) + \xi_u(t)\ubf(t),
        \label{eq:genAMPb} \\
    \dbf(t) &=& \Sbf^T\ubf(\tp1) + \xi_v(t)\vbf(t) \label{eq:genAMPd}
\eeqa
\end{subequations}
where
\begin{subequations} \label{eq:genAMPuv}
\beqa
    u_i(\tp1) &=& H_u(t,b_i(t),u_{0i},\nu_u(t)), \\
    v_j(\tp1) &=& H_v(t,d_j(t),v_{0j},\nu_v(t)),
\eeqa
\end{subequations}
and $\xi_v(t)$ and $\xi_u(t)$ are scalar step sizes given by
\begin{subequations} \label{eq:genAMPDeriv}
\beqa
    \xi_v(t) &=& -\frac{1}{m} \sum_{i=1}^m
        \frac{\partial}{\partial b_i}H_u(t,b_i(t),u_{0i},\nu_u(t)) \\
    \xi_u(\tp1) &=& -\frac{1}{m} \sum_{j=1}^n  \frac{\partial}{\partial d_j}
    H_v(t,d_j(t),v_{0j},\nu_v(t)). \hspace{0.7cm}
        \label{eq:genAMPmuu}
\eeqa
\end{subequations}
The recursions \eqref{eq:genAMPRec} to \eqref{eq:genAMPDeriv}
are identical to the recursions analyzed in \citep{BayatiM:11},
except for the introduction of the parameters $\nu_u(t)$
and $\nu_v(t)$.  We will call these parameters \emph{adaptation
parameters} since they enable the functions $H_u(\cdot)$ and $H_v(\cdot)$
to have some data dependence, not explicitly considered in
\citep{BayatiM:11}.  Similar to the selection of the parameters
$\lambda_u(t)$ and $\lambda_v(t)$ in \eqref{eq:lamRankOne},
we assume that, in each iteration $t$,
the adaptation parameters are selected by functions of the form,
\beq \label{eq:genNu}
    \nu_u(t) = \frac{1}{n} \sum_{j=1}^n \phi_u(t,v_{0j},v_j(t)),
    \quad        
    \nu_v(t) = \frac{1}{m} \sum_{i=1}^m \phi_v(t,u_{0i},u_i(\tp1))
\eeq
where $\phi_u(\cdot)$ and $\phi_v(\cdot)$ are
(possibly vector-valued) pseudo-Lipschitz continuous
of order $p$ for some $p > 1$.
Thus, the values of $\nu_u(t)$ and $\nu_v(t)$ depend on the
outputs $\vbf(t)$ and $\ubf(\tp1)$.
Note that in equations \eqref{eq:genAMPuv} to \eqref{eq:genNu},
$d_i$, $u_{0i}$, $b_j$ and $v_{0j}$ are the components
of the vectors $\dbf$, $\ubf_0$, $\bbf$ and $\vbf_0$, respectively.
The algorithm is initialized with $t=0$, $\xi_u(0)=0$ and
some vector $\vbf(0)$.

Now, similar to Section~\ref{sec:asymAnal}, consider
a sequence of random realizations of the
parameters indexed by the input dimension $n$.  For each $n$, we
assume that the output dimension $m=m(n)$
is deterministic and scales linearly as in \eqref{eq:betaDef}
for some $\beta \geq 0$.
Assume that the transform matrix $\Sbf$ has i.i.d.\ Gaussian
components $s_{ij} \sim {\cal N}(0,1/m)$.
Also assume that the empirical limits in
\eqref{eq:thetauvLimInit} hold with bounded moments of order $2p-2$
for some limiting random variables $(U_0,U(0))$ and $(V_0,V(0))$.
We will also assume the following continuity assumptions on $H_u(\cdot)$
and $H_v(\cdot)$:

\begin{assumption} \label{as:Huv} The function $H_u(t,b,u_0,\nu_u)$
satisfies the following continuity conditions:
\begin{itemize}
\item[(a)]  For every $\nu_u$ and $t$, $H_u(t,b,u_0,\nu_u)$
and its derivative $\partial H_u(t,b,u_0,\nu_u)/\partial b$
are Lipschitz continuous in $b$ and $u_0$ for some Lipschitz constant
that is continuous in $\nu_u$; and
\item[(b)]  For every $\nu_u$ and $t$, $H_u(t,b,u_0,\nu_u)$
and $\partial H_u(t,b,u_0,\nu_u)/\partial b$ are
is continuous at $\nu_u$ uniformly over $(b,u_0)$.
\end{itemize}
The function $H_v(t,d,v_0,\nu_v)$
satisfies the analogous continuity assumptions
in $d$, $v_0$ and $\nu_v$.
\end{assumption}

\medskip
Under these assumption,  we will show, as in Section~\ref{sec:asymAnal},
 that for any fixed iteration $t$, the
sets $\theta_u(t)$ and $\theta_v(t)$ in \eqref{eq:thetauvi}
converge empirically to the limits \eqref{eq:thetauvLim} for some
random variable pairs $(U_0,U(t))$ and $(V_0,V(t))$.
The random variable $U_0$ is identical to the
variable in the limit \eqref{eq:thetauvLimInit} and, for $t \geq 0$,
$U(t)$ is given by
\beq \label{eq:UtGen}
    U(\tp1) = H_u(t,B(t),U_0,\nubar_u(t)), \quad
    B(t) \sim {\mathcal N}(0,\tau^b(t)),
\eeq
for some deterministic constants $\nubar_v(t)$ and $\tau^b(t)$
that will be defined below.
Similarly, the random variable $V_0$ is identical to the
variable in the limit \eqref{eq:thetauvLimInit} and, for $t \geq 0$,
$V(t)$ is given by
\beq    \label{eq:VtGen}
    V(\tp1) = H_v(t,D(t),V_0,\nubar_v(t)), \quad
       D(t) \sim {\mathcal N}(0,\tau^d(t)),
\eeq
for some constants $\nubar_v(t)$ and $\tau^d(t)$, also defined below.

The constants $\tau^b(t)$, $\tau^d(t)$, $\nubar_u(t)$
and $\nubar_v(t)$ can be computed
recursively through the following state evolution equations
\begin{subequations} \label{eq:SEGen}
\beqa
    \tau^d(t) &=&
        \Exp\left[ U(\tp1)^2 \right], 
    \quad \tau^b(t) = \beta
        \Exp\left[ V(t)^2 \right] \label{eq:taub} \\
    \nubar_u(t) &=& \Exp\left[\phi_u(t,V_0,V(t))\right], 
    \quad
    \nubar_v(t) = \Exp\left[\phi_v(t,U_0,U(\tp1))\right] \label{eq:nuv}
\eeqa
\end{subequations}
where the expectations are over the random variables
$U(t)$ and $V(t)$ above and initialized with
\beq \label{eq:SEGenInit}
    \tau^b(0) := \beta \Exp\left[ V(0)^2 \right].
\eeq
With these definitions, we can now state the adaptive version of the
result from Bayati and Montanari \citep{BayatiM:11}.
Although the full proof requires that $p=2$,
much of the proof is valid for $p\geq 2$.  Hence, where possible,
we provide the steps for the general $p$ case.
\medskip

\begin{theorem} \label{thm:SEGen}
Consider the recursion in \eqref{eq:genAMPRec} to \eqref{eq:genNu}
satisfying the above assumptions for the case when $p=2$.
Then, for any fixed
iteration number $t$, the sets
$\theta_u(t)$ and $\theta_v(t)$ in \eqref{eq:thetauvi}
converge empirically to the limits \eqref{eq:thetauvLim}
with bounded moments of order $p=2$ to the
random variable pairs $(U_0,U(t))$ and $(V_0,V(t))$
described above.
\end{theorem}

\begin{proof}  We use an asterisk superscript
to denote the outputs of the \emph{non-adaptive}
version of the recursions \eqref{eq:genAMPRec} to \eqref{eq:genNu}.
That is, quantities such as
$\ubf^*(t), \vbf^*(t), \bbf^*(t), \dbf^*(t), \ldots$, will represent
the outputs generated by
recursions \eqref{eq:genAMPRec} to \eqref{eq:genNu} with the
same initial conditions ($\vbf^*(0)=\vbf(0)$ and $\xi_u(0)=\xi_u^*(0)=0$),
but in \eqref{eq:genAMPuv} and \eqref{eq:genAMPDeriv},
$\nu_u(t)$ and $\nu_v(t)$ are replaced by their deterministic
limits $\nubar_u(t)$ and $\nubar_v(t)$.
Therefore, the
\begin{subequations} \label{eq:genAMPuvdet}
\beqa
    u^*_i(\tp1) &=& H_u(t,b^*_i(t),u^*_{0i},\nubar_u(t)), \\
    v^*_j(\tp1) &=& H_v(t,d^*_j(t),v^*_{0j},\nubar_v(\tp1)),
\eeqa
\end{subequations}
and
\begin{subequations} \label{eq:genAMPDerivdet}
\beqa
    \xi^*_v(t) &=& -\frac{1}{m} \sum_{i=1}^m
        \frac{\partial}{\partial b_i}H_u(t,b^*_i(t),u^*_{0i},\nubar_u(t)) \\
    \xi^*_u(\tp1) &=& -\frac{1}{m} \sum_{j=1}^n  \frac{\partial}{\partial d_j}
    H_v(t,d^*_j(t),v^*_{0j},\nubar_v(t)). \hspace{0.7cm}
\eeqa
\end{subequations}
Now, Bayati and Montanari's result in \citep{BayatiM:11} shows that
this non-adaptive algorithm satisfies the required properties.  That is,
the following limits hold with bounded moments of order $p$,
\begin{subequations} \label{eq:thetauvLimNA}
\beqa
    \lim_{n \arr \infty} \{(u_{0i},u^*_i(t)), i=1,\ldots,m\} &=& (U_0,U(t)) \\
    \lim_{n \arr \infty} \{(v_{0j},v^*_j(t)), j=1,\ldots,n\} &=& (V_0,V(t)).
\eeqa
\end{subequations}
So, the limits \eqref{eq:thetauvLim} will be shown if we can prove
the following limits hold almost surely for all $t$:
\beq \label{eq:uvlimd}
    \lim_{n \arr \infty} \frac{1}{m} \|\ubf(t)-\ubf^*(t)\|_p^p = 0,
    \quad 
    \lim_{n \arr \infty} \frac{1}{n} \|\vbf(t)-\vbf^*(t)\|_p^p = 0,
\eeq
where $\|\cdot\|_p$ is the $p$-norm.
In the course of proving \eqref{eq:uvlimd}, we will also
show the following limits hold almost surely,
\begin{subequations} \label{eq:otherlimd}
\beqa
    \lim_{n \arr \infty} \frac{1}{m} \|\bbf(t)-\bbf^*(t)\|_p^p &=& 0
        \label{eq:blimd} \\
    \lim_{n \arr \infty} \frac{1}{n} \|\dbf(t)-\dbf^*(t)\|_p^p &=& 0
        \label{eq:dlimd} \\
    \lim_{n \arr \infty} |\xi_u(t)-\xi^*_u(t)| &=& 0 \label{eq:xiulimd} \\
    \lim_{n \arr \infty} |\xi_v(t)-\xi^*_v(t)| &=& 0 \label{eq:xivlimd} \\
    \lim_{n \arr \infty} \nu_u(t) &=& \nubar_u(t) \label{eq:nuulimd} \\
    \lim_{n \arr \infty} \nu_v(t) &=& \nubar_v(t) \label{eq:nuvlimd}
\eeqa
\end{subequations}
The proof of the limits \eqref{eq:uvlimd} and \eqref{eq:otherlimd}
can be demonstrated via induction on $t$ with the following
straightforward (but somewhat tedious) continuity argument:

To begin the induction argument, first note that
the non-adaptive algorithm has the same initial condition as
the adaptive algorithm.  That is,
$\vbf^*(0)=\vbf(0)$ and $\xi_u(0)=\xi_u^*(0)=0$.  Also, since
$\xi_u(0)=\xi_u^*(0)=0$, from \eqref{eq:genAMPb},
the initial value of $\ubf(t)$ does
not matter.  So, without loss of generality, we can assume that the
initial condition satisfies $\ubf(0)=\ubf^*(0)$.  Thus, the limits
\eqref{eq:uvlimd} and
\eqref{eq:xiulimd} hold for $t=0$.

We now proceed by induction.
Suppose that the limits \eqref{eq:uvlimd} and
\eqref{eq:xiulimd} hold almost surely for some $t \geq 0$.  Since $\Sbf$
has i.i.d.\ components with zero mean and variance $1/m$,
by the Marceko-Pastur Theorem \citep{MarcenkoP:67},
the maximum singular value of $\Sbf$ is bounded.  For $p=2$,
the maximum singular value is the $p$-norm operator norm,
and therefore, there exists a
constant $C_S > 0$ such that
\beq \label{eq:Snorm}
    \limsup_{n \arr \infty} \|\Sbf\|_p \leq C_S, \quad
    \limsup_{n \arr \infty} \|\Sbf^T\|_p \leq C_S.
\eeq
Substituting the bound \eqref{eq:Snorm} into \eqref{eq:genAMPb},
we obtain
\beqa
    \lefteqn{ \|\bbf(t)-\bbf^*(t)\|_p \leq \|\Sbf\|_p\|\vbf(t)-\vbf^*(t)\|_p
         + |\xi_u(t)|\|\ubf(t)-\ubf^*(t)\|_p +
         |\xi_u(t)-\xi^*_u(t)|\|\ubf^*(t)\|_p } \nonumber \\
         & \leq& C_S\|\vbf(t)-\vbf^*(t)\|_p + |\xi_u(t)|\|\ubf(t)-\ubf^*(t)\|_p
         + |\xi_u(t)-\xi^*_u(t)|\|\ubf^*(t)\|_p.
         \hspace{2cm} \label{eq:blimbnd1}
\eeqa
Now, since $p \geq 1$, we have that for any positive numbers
$a$ and $b$,
\beq \label{eq:abpIneq}
    (a+b)^p \leq 2^{p-1}(a^p + b^p).
\eeq
Applying \eqref{eq:abpIneq} into \eqref{eq:blimbnd1},
and the fact that $\lim_n m/n = \beta$, we obtain that
\beqa
    \lefteqn{
     \frac{1}{m}\|\bbf(t)-\bbf^*(t)\|^p_p \leq
         C'_S\Bigl[ \frac{1}{n}\|\vbf(t)-\vbf^*(t)\|^p_p
         } \nonumber \\
        & &  + \frac{|\xi_u(t)|^p}{m}\|\ubf(t)-\ubf^*(t)\|^p_p
      +\frac{|\xi_u(t)-\xi^*_u(t)|^p}{m}
            \|\ubf^*(t)\|^p_p\Bigr], \label{eq:blimbnd2}
\eeqa
for some other constant $C'_S > 0$.
Now, since $\ubf^*(t)$ is the output of the non-adaptive
algorithm it satisfies the limit
\beq \label{eq:ubndNA}
    \lim_{n \arr \infty} \frac{1}{m}\|\ubf^*(t)\|_p^p
    = \lim_{n \arr \infty} \frac{1}{m}\sum_{i=1}^m|u_i^*(t)|_p^p
    = \Exp|U(t)|^p < \infty.
\eeq
Substituting the bound \eqref{eq:ubndNA} along with
induction hypotheses, \eqref{eq:uvlimd} and
\eqref{eq:xiulimd} into \eqref{eq:blimbnd2}
shows \eqref{eq:blimd}.

Next, to prove the limit \eqref{eq:nuulimd},
first observe that since $\phi_u(\cdot)$ is pseudo-Lipschitz
continuous of order $p$, we have that
$\nubar_u(t)$ in \eqref{eq:nuv} can be replaced by
the limit of the empirical means
\beq \label{eq:nuuLimNA}
    \nubar_u(t) = \lim_{n \arr\infty} \frac{1}{n}
        \sum_{j=1}^n \phi_u(t,v_{0j},v_j^*(t)),
\eeq
where the limit holds almost surely.
Combining \eqref{eq:nuuLimNA} with the \eqref{eq:genNu},
\beqa
    \limsup_{n \arr \infty} |\nu_u(t)-\nubar_u(t)|
    \leq
    \limsup_{n \arr\infty} \frac{1}{n}
        \sum_{j=1}^n |\phi_u(t,v_{0j},v_j^*(t)) -\phi_u(t,v_{0j},v_j^*(t))|
        \label{eq:nuuDiff1}
\eeqa
Applying the fact that
$\phi_u(\cdot)$ is pseudo-Lipschitz continuous of order $p$
to \eqref{eq:nuuDiff1}, we obtain that
there exists a constant $L_v > 0$ such that
\beqa
    \lefteqn{ \limsup_{n \arr \infty} |\nu_u(t)-\nubar_u(t)|
    \leq \limsup_{n \arr \infty}
    \frac{L_v}{n}\sum_{j=1}^n\Bigl[ |v_j(t)|^{p-1}
    + |v_j^*(t)|^{p-1}\Bigr]
    |v_j(t)-v_j^*(t)| } \nonumber \\
    &\leq& \limsup_{n \arr \infty}
    L_v\Bigl[ \Bigl(\frac{1}{n}\|\vbf(t)\|_{p}^{p}\Bigr)^{(p-1)/p}
    + \Bigl(\frac{1}{n}\|\vbf^*(t)\|_{p}^{p}\Bigr)^{(p-1)/p}
    \Bigr]\Bigl(\frac{1}{n}\|\vbf(t)-\vbf^*(t)\|^p_p\Bigr)^{1/p},
      \label{eq:nuuDiff2}
\eeqa
where the last step is due to H{\"o}lder's inequality with the
exponents
\[
    \frac{p-1}{p} + \frac{1}{p} = 1.
\]
Now, similar to the proof of \eqref{eq:ubndNA} one can show that
the non-adaptive output satisfies the limit
\beq \label{eq:vbnd1}
    \lim_n \frac{1}{n} \|\vbf^*(t)\|_p^p = \Exp|V(t)|^p < \infty
\eeq
Also, from the induction hypothesis \eqref{eq:uvlimd},
it follows that the non-adaptive output must satisfy the same
limit
\beq \label{eq:vbnd2}
    \lim_n \frac{1}{n} \|\vbf(t)\|_p^p = \Exp|V(t)|^p < \infty.
\eeq
Applying the bounds \eqref{eq:vbnd1} and \eqref{eq:vbnd2}
and the limit \eqref{eq:nuuDiff2} shows that \eqref{eq:nuulimd} holds
almost surely.

Now the limit in \eqref{eq:nuulimd} and the second limit in
\eqref{eq:uvlimd}
together with the continuity conditions on $H_u(\cdot)$
in Assumption \ref{as:Huv}
show that the first limit in \eqref{eq:uvlimd} holds almost surely for $t+1$
and \eqref{eq:xivlimd} holds almost surely for $t$.
Using \eqref{eq:genAMPd}, the proof of the limit \eqref{eq:dlimd}
is similar to the proof of \eqref{eq:blimd}.
These limits in turn show that the second limit in \eqref{eq:uvlimd} and
the limits in \eqref{eq:xiulimd}
hold almost surely for $t+1$.  We have thus shown
that if \eqref{eq:uvlimd} and
\eqref{eq:xiulimd} hold almost surely for some $t$, they hold for $t+1$.
Thus, by induction they hold for all $t$.
Finally, applying the limits \eqref{eq:thetauvLimNA},
\eqref{eq:uvlimd}
and a continuity argument shows that the desired limits
\eqref{eq:thetauvLim} hold almost surely.

\end{proof}

\subsection{Proof of Theorem~\ref{thm:rankOneLim}} \label{sec:rankOneLimPf}

The theorem directly follows from the adaptive
Bayati--Montanari recursion theorem,
Theorem~\ref{thm:SEGen} above, with some change of variables.
Specifically, let
\beq \label{eq:Sdef}
    \Sbf = \frac{1}{\sqrt{m\tau_w}} \Wbf,
\eeq
where $\Wbf$ is the Gaussian noise in the rank one model in
Assumption~\ref{as:rankOne}(b).  Since $\Wbf$ has i.i.d.\ components
with Gaussian distributions ${\mathcal N}(0,\tau_w)$, the
components of $\Sbf$ will be i.i.d.\ with distributions
${\mathcal N}(0,1/m)$.

Now, using the rank one model for $\Abf$ in \eqref{eq:ArankOne}
\beq     \label{eq:Avpf}
    \Abf\vbf(t) =  \ubf_0\vbf_0^T\vbf(t) + \sqrt{m}\Wbf\vbf(t)
    = n\alpha_{v1}(t)\ubf_0 + \sqrt{m}\Wbf\vbf(t),
\eeq
where the last step is from the definition of $\alpha_{v1}(t)$ in \eqref{eq:alphauv}.
Substituting \eqref{eq:Avpf} into the the update rule for $\pbf(t)$
in line \ref{line:pt} of Algorithm \ref{algo:iterMax},
we obtain
\beqa
    \lefteqn{ \pbf(t) =
        (1/m)\Abf(t)\vbf(t) + \mu_u(t)\ubf(t) }  \nonumber \\
    &=& (1/\sqrt{m})\Wbf\vbf(t)
        + \beta \alpha_{v1}(t)\ubf_0 +  \mu_u(t)\ubf(t). \label{eq:Wvpf}
\eeqa
Note that we have used the fact that $\beta=n/m$.
Hence, if we define
\beq \label{eq:bdefpf}
    \bbf(t) = \frac{1}{\sqrt{\tau_w}}(\pbf(t)-\beta\alpha_{v1}(t)\ubf_0),
\eeq
then \eqref{eq:Sdef} and \eqref{eq:Wvpf} show that
\beq \label{eq:bSpf}
    \bbf(t) = \Sbf(t)\vbf(t) + \xi_u(t)\ubf(t), \quad
    \xi_u(t) = \mu_u(t)/\sqrt{\tau_w}.
\eeq

Similarly, one can show that if we define
\beq \label{eq:ddefpf}
    \dbf(t) = \frac{1}{\sqrt{\tau_w}}(\qbf(t)-\alpha_{u1}(t)\vbf_0),
\eeq
then
\beq \label{eq:dSpf}
    \dbf(t) = \Sbf(t)^T\ubf(\tp1) + \xi_v(t)\vbf(t), \quad
    \xi_v(t) = \mu_v(t)/\sqrt{\tau_w}.
\eeq

Next define the adaptation functions
\beq \label{eq:phiuvpf}
    \phi_u(t,v_0,v) := (vv_0, \phi_{\lambda u}(t,v_0,v)), \quad
    \phi_v(t,u_0,u) := (uu_0, \phi_{\lambda v}(t,u_0,u))
\eeq
which are the adaptation functions in \eqref{eq:lamRankOne} with additional
components for the second-order statistics $uu_0$ and $vv_0$.
Since $\phi_{\lambda u}(t,\cdot)$ and $\phi_{\lambda v}(t,\cdot)$
are pseudo-Lipschitz of order $p$, so are
$\phi_{u}(t,\cdot)$ and $\phi_{v}(t,\cdot)$.
Taking the empirical
means over each of the two components of $\phi_u(\cdot)$ and $\phi_v(\cdot)$,
and applying \eqref{eq:lamRankOne} and \eqref{eq:alphauv}, we see that
if $\nu_u(t)$ and $\nu_v(t)$ are defined as in \eqref{eq:genNu},
\begin{subequations}  \label{eq:nuuvpf}
\beqa
    \nu_u(t) &=& \frac{1}{n}\sum_{j=1}^n \phi_v(t,v_{0j},v_j(t)) =
        (\alpha_{v1}(t),\lambda_u(t)) \\
    \nu_v(t) &=& \frac{1}{m}\sum_{i=1}^m \phi_u(t,u_{0i},u_i(\tp1))
     = (\alpha_{u1}(\tp1),\lambda_v(t))
\eeqa
\end{subequations}
Therefore, $\nu_u(t)$ and $\nu_v(t)$ are vectors containing the parameters
$\lambda_u(t)$ and $\lambda_v(t)$ for the factor selection functions
in lines \ref{line:Guopt} and \ref{line:Gvopt} of Algorithm \ref{algo:iterMax}
as well as the second-order statistics $\alpha_{u1}(t)$ and $\alpha_{v1}(t)$.
Now, for $\nu_u = (\alpha_{v1},\lambda_u)$ and $\nu_v = (\alpha_{u1},\lambda_v)$
define the scalar functions
\begin{subequations} \label{eq:Huvdefpf}
\beqa
    H_u(t,b,u_0,\nu_u) &:=& G_u(\sqrt{\tau_w} b + \beta\alpha_{v1}u_0,\lambda_u)
        \label{eq:Hudefpf} \\
    H_v(t,d,v_0,\nu_v) &:=& G_v(\sqrt{\tau_w} d + \alpha_{u1}v_0,\lambda_v).
        \label{eq:Hvdefpf}
\eeqa
\end{subequations}
Since $G_u(p,\lambda_u)$ and $G_v(q,\lambda_v)$ satisfy the continuity
conditions in Assumption \ref{as:rankOne}(c),
$H_u(t,b,u_0,\lambda_u)$ and $H_v(t,d,v_0,\lambda_v)$ satisfy Assumption
\ref{as:Huv}.
In addition, the componentwise separability assumption in \eqref{eq:GuvSep}
implies that the updates in lines
\ref{line:Guopt} and \ref{line:Gvopt} of Algorithm~\ref{algo:iterMax}
can be rewritten as
\beq \label{eq:Guvpf}
    u_i(\tp1) = G_u(p_i(t),\lambda_u(t)), \quad
    v_j(\tp1) = G_v(d_j(t),\lambda_v(t)).
\eeq
Thus, combining \eqref{eq:Huvdefpf} and \eqref{eq:Guvpf}
with the definitions of $\bbf(t)$ and $\dbf(t)$ in
\eqref{eq:bdefpf} and \eqref{eq:ddefpf}, we obtain
\beq  \label{eq:Huvpf}
    u_i(\tp1) = H_u(t,b_i(t),u_{0i},\nu_u(t)), \quad
    v_j(\tp1) = H_v(t,d_j(t),v_{0j},\nu_v(t)).
\eeq

Next observe that
\beqa
    \lefteqn{ \xi_u(\tp1) \stackrel{(a)}{=} \mu_v(\tp1)/\sqrt{\tau_w}
    \stackrel{(b)}{=} -\frac{\sqrt{\tau_w}}{m}\sum_{j=1}^n
    \frac{\partial}{\partial q_j}G_v(q_j(t),\lambda_v(t)) } \nonumber \\
     &\stackrel{(c)}{=}& -\frac{1}{m}\sum_{j=1}^n
    \frac{\partial}{\partial d_j}H_v(t,d_j(t),\nu_v(t))
    \hspace{3cm} \label{eq:xiuH}
\eeqa
where (a) follows from the definition of $\xi_u(t)$ in \eqref{eq:bSpf};
(b) is the setting for $\mu_u(\tp1)$ in line~\ref{line:muu};
and (c) follows from the definition of $H_v(t,d)$ in \eqref{eq:Hvdefpf}.
Similarly, one can show that
\beq \label{eq:xivH}
    \xi_v(t) = -\frac{1}{m}\sum_{i=1}^m
    \frac{\partial}{\partial b_i}H_u(t,b_i(t),\nu_u(t)).
\eeq

Equations \eqref{eq:bSpf}, \eqref{eq:dSpf}, \eqref{eq:nuuvpf},
\eqref{eq:Huvdefpf}, \eqref{eq:xiuH} and \eqref{eq:xivH} exactly match the
recursions in equations \eqref{eq:genAMPRec} to \eqref{eq:genNu}.
Therefore, Theorem~\ref{thm:SEGen} shows that
the limits \eqref{eq:thetauvLim} hold in a sense that
the sets $\theta_u(t)$ and $\theta_v(t)$ converge
empirically with bounded moments of order $p$.

We next show that the limits $U(t)$ and $V(t)$
on the right-hand side of \eqref{eq:thetauvLim} match the descriptions in
\eqref{eq:USE} and \eqref{eq:VSE}.  First, define $\alphabar_{u1}(t)$
and $\alphabar_{v1}(t)$ in \eqref{eq:alphaSE} and $\lambdabar_u(t)$
and $\lambdabar_v(t)$ as in \eqref{eq:lambdaSE}.  Then,
from \eqref{eq:phiuvpf}, the expectations $\nubar_u(t)$ and
$\nubar_v(t)$ in \eqref{eq:nuv} are given by
\beq \label{eq:nubarpf}
    \nubar_u(t) = (\alphabar_{v1}(t),\lambdabar_u(t)), \quad
    \nubar_v(t) = (\alphabar_{u1}(\tp1),\lambdabar_v(t)).
\eeq
Using \eqref{eq:UtGen}, \eqref{eq:Hudefpf} and \eqref{eq:nubarpf},
we see that
\beqa
    U(\tp1) &=& H_u(t,B(t),U_0,\nubar_u(t)) \nonumber \\
        &=& G_u\bigl(t, \beta\alphabar_{v1}(t)U_0 + \sqrt{\tau_w} B(t),
        \lambdabar_u(t) \bigr),
            \label{eq:UtGB}
\eeqa
where $B(t) \sim {\mathcal N}(0,\tau^b(t))$.
Therefore,  if we let $Z_u(t) = \sqrt{\tau_w} B(t)$,
then $Z_u(t)$ is zero mean Gaussian with variance
\[
    \Exp\left[Z_u^2(t)\right] = \tau_w \tau^b(t)
    \stackrel{(a)}{=} \beta \tau_w \Exp[V(t)^2]
    \stackrel{(b)}{=} \beta \tau_w \alphabar_{v0}(t),
\]
where follows from \eqref{eq:taub} and (b) follows from
the definition of $\alphabar_{v0}(t)$ in \eqref{eq:alphaSE}.
Substituting  $Z_u(t) = \sqrt{\tau_w} B(t)$ into \eqref{eq:UtGB} we obtain
the model for $U(\tp1)$ in \eqref{eq:USE}.
Similarly, using \eqref{eq:VtGen} and \eqref{eq:Hvdefpf},
we can obtain the model $V(\tp1)$ in \eqref{eq:VSE}.
Thus, we have proven that the random variables $(U_0,U(t))$
and $(V_0,V(t))$ are described by \eqref{eq:USE} and \eqref{eq:VSE},
and this completes the proof.

\subsection{Proof of Theorem~\ref{thm:linConv}} \label{sec:linConvPf}

The theorem is proven by simply evaluating the second order statistics.
We begin with $\alphabar_{u1}(\tp1)$:
\beqa
    \lefteqn{ \alphabar_{u1}(\tp1) \stackrel{(a)}{=}
    \Exp[ U_0U(\tp1) ] \stackrel{(b)}{=} \lambdabar_u(t) \Exp[ U_0P(t) ] } \nonumber \\
    &\stackrel{(c)}{=}& \lambdabar_u(t) \Exp\left[ U_0(\beta\alphabar_{v1}(t)U_0+Z_u(t)) \right]
    \stackrel{(d)}{=}  \lambdabar_u(t)\beta\tau_u\alphabar_{v1}(t) \label{eq:alphau1Lin}
\eeqa
where (a) is the definition in \eqref{eq:alphaSE};
(b) follows from \eqref{eq:USE} and \eqref{eq:Guvlin};
(c) follows from \eqref{eq:USE}; and (d) follows from the independence of
$Z_u(t)$ and $U_0$ and the definition of $\tau_u$ in \eqref{eq:tauuv}.
Similarly, one can show that
\beq \label{eq:alphau0Lin}
 \alphabar_{u0}(\tp1) = \lambdabar_u(t)^2\left[
    \beta\tau_u\alphabar^2_{v1}(t) + \beta\tau_w\alphabar_{v0}(t)\right].
\eeq
Substituting \eqref{eq:alphau1Lin} and \eqref{eq:alphau0Lin} into \eqref{eq:corruLim},
we obtain the asymptotic correlation coefficient
\beqa
    \lefteqn{\rho_u(\tp1) = \frac{\lambdabar^2_u(t)\beta^2\tau_u^2\alphabar^2_{v1}(t)}
       {\lambdabar^2_u(t)\left[
    \beta\tau_u\alphabar^2_{v1}(t) + \beta\tau_w\alphabar_{v0}(t)\right]\tau_u} }
    \nonumber \\
    &=&  \frac{\beta\tau_u\alphabar_{v1}^2(t)}
       {\tau_u\alphabar^2_{v1}(t) + \tau_w\alphabar_{v0}(t)}
    =  \frac{\beta\tau_u\tau_v\rho_{v}(t)}
       {\tau_u\tau_v\rho_{v}(t) + \tau_w},
        \nonumber
\eeqa
where the last step follows from \eqref{eq:corrvLim}.
This proves the first equation in  \eqref{eq:corrLinSE}.

A similar set of calculations shows that
\begin{subequations}
\beqa
    \alphabar_{v1}(\tp1) &=&
    \lambdabar_v(t)\tau_v\alphabar_{u1}(\tp1) \\
     \alphabar_{v0}(\tp1) &=&  \lambdabar_v(t)^2\left[
    \tau_v\alphabar^2_{u1}(\tp1) + \tau_w\alphabar_{u0}(t)\right].
\eeqa
\end{subequations}
Applying these equations into \eqref{eq:corruLim} and \eqref{eq:corrvLim},
we obtain the recursion \eqref{eq:corrLinSE}.
Hence, we have proven part (a) of the theorem.

For part (b), we need the following simple lemma.

\begin{lemma} \label{lem:monoiter}
Suppose that $H:[0,1] \arr [0,1]$ is continuous and monotonically increasing,
and $x(t)$ is a sequence satisfying the recursive relation
\beq \label{eq:monoiter}
    x(t+1) = H(x(t)),
\eeq
for some initial condition $x(0) \in [0,1]$.  Then, either $x(t)$ monotonically
increases or decreases to some $x^* = H(x^*)$.
\end{lemma}
\begin{proof} This can be proven similar to \citep[lemma 7]{Rangan:10arXiv}.
\end{proof}

To apply Lemma \ref{lem:monoiter}, observe that the recursions \eqref{eq:corrLinSE}
show that
\[
     \rho_u(\tp1) = \frac{\beta \tau_u\tau_v \rho_v(t)}{
    \beta \tau_u\tau_v \rho_v(t) + \tau_w}
    = \frac{\beta \tau_u^2\tau_v^2 \rho_u(t)}{
    \beta \tau_u^2\tau_v^2 \rho_u(t) + \tau_w(\tau_u\tau_v \rho_u(t) + \tau_w)}.
\]
So, if we define
\beq \label{eq:HdefLin}
    H(\rho_u) := \frac{\beta \tau_u^2\tau_v^2 \rho_u}{
    (\beta \tau_u^2\tau_v^2 + \tau_w\tau_u\tau_v)\rho_u + \tau_w^2},
\eeq
then it follows from \eqref{eq:corrLinSE} that
$\rho_u(\tp1) = H(\rho_u(t))$ for all $t$.
By taking the derivative, it can be checked that
$H(\rho_u)$ is monotonically increasing.
It follows from Lemma \ref{lem:monoiter} that $\rho_u(t) \arr \rho_u^*$
for some fixed point $\rho_u^* = H(\rho_u^*)$ with
$\rho_u^* \in [0,1]$.

Now, there are only two
fixed point solutions to $\rho_u^* = H(\rho_u^*)$: $\rho_u^* = 0$ and
\beq \label{eq:corruLinLimitPos}
    \rho_u^* = \frac{\beta\tau_u^2\tau_v^2 - \tau_w^2}{
    \tau_u\tau_v(\beta\tau_u\tau_v + \tau_w)}.
\eeq
When
\beq \label{eq:snrMinPf}
    \beta\tau_u^2\tau_v^2 \leq \tau_w^2,
\eeq
then $\rho_u^*$ in \eqref{eq:corruLinLimitPos} is not positive,
so the only fixed solution in $[0,1]$ is $\rho_u^* = 0$.  Therefore, when
\eqref{eq:snrMinPf} is not satisfied, $\rho_u(t)$ must converge to the
zero fixed point: $\rho_u(t) \arr 0$.

Now, suppose that \eqref{eq:snrMinPf} is satisfied.
In this case, we claim that $\rho_u(t) \arr \rho_u^*$ where $\rho_u^*$ is
in \eqref{eq:corruLinLimitPos}.  We prove this claim by contradiction and
suppose, instead,  that $\rho_u(t)$ converges to the other fixed point:
$\rho_u(t) \arr 0$.
Since Lemma \ref{lem:monoiter} shows that $\rho_u(t)$ must be either
monotonically increasing or decreasing, the only way $\rho_u(t) \arr 0$
is that $\rho_u(t)$ monotonically decreases to zero.
But, when \eqref{eq:snrMinPf} is satisfied, it can be checked that for
$\rho_u(t)$ sufficiently small and positive, $\rho_u(\tp1) > H(\rho_u(t))$.
This contradicts the fact that $\rho_u(t)$ is monotonically decreasing,
and therefore, $\rho_u(t)$ must converge to the other fixed point $\rho_u^*$
in \eqref{eq:corruLinLimitPos}.

Hence, we have shown that when \eqref{eq:snrMinPf} is not satisfied,
$\rho_u(t) \arr 0$, and when \eqref{eq:snrMinPf} is satisfied
$\rho_u(t) \arr \rho_u^*$ in \eqref{eq:corruLinLimitPos}.  This is equivalent to
the first limit in \eqref{eq:corrLinLimit}.  The second limit \eqref{eq:corrLinLimit}
is proved similarly.

\subsection{Proof of Theorem \ref{thm:SEmmse}} \label{sec:SEmmsePf}

Similar to the proof of Theorem \ref{thm:linConv}, we begin by computing
the second-order statistics of $(U_0,U(t))$.
Since $U(t) = \Exp[U_0 \MID P(\tm1)]$, $U(t)$ must be uncorrelated with the
error:  $U(t)(U_0-U(t))=0$.  Hence,
\beq
 \alphabar_{u0}(t) - \alphabar_{u1}(t) =
    \Exp[U(t)U_0 - U(t)U(t)] = 0,
    \label{eq:aueqmmse}
\eeq
and therefore $\alphabar_{u0}(t) = \alphabar_{u1}(t)$.
Now, consider the measurement $P(t)$ in \eqref{eq:USE}.  The
SNR in this channel is
\beq \label{eq:snru}
    \eta_u(t) = \frac{\beta^2\alphabar^2_{v1}(t)}{\beta\tau_w\alphabar_{v0}(t)}
    = \frac{\beta \tau_v\rho_v(t)}{\tau_w}.
\eeq
Since $U(\tp1)$ is the conditional expectation of $U_0$ given $P(t)$,
the mean-squared error is given by $\Ecal_u(\eta_u(t))$ defined in \eqref{eq:mmseuv}.
Therefore,
\beqa
     \Ecal_u(\eta_u(t)) = \Exp[U(\tp1)-U_0]^2
   \stackrel{(a)}{=} \alphabar_{u0}(\tp1)-2\alphabar_{u1}(\tp1) + \tau_u
    \stackrel{(b)}{=} \tau_u - \alphabar_{u0}(\tp1),  \label{eq:mmseau0}
\eeqa
where (a) follows from expanding the square and substituting in the definitions in
\eqref{eq:alphaSE} and \eqref{eq:tauuv}; and
(b) follows from the fact that $\alphabar_{u0}(\tp1) = \alphabar_{u1}(\tp1)$
proven above.  We have thus proven that
\beq \label{eq:alphaummse}
    \alphabar_{u0}(\tp1) = \alphabar_{u1}(\tp1) = \tau_u - \Ecal_u(\eta_u(t)).
\eeq
Therefore, the asymptotic correlation coefficient is given by
\beqa
\rho_u(\tp1) \stackrel{(a)}{=} \frac{\alphabar_{u1}^2(\tp1)}{
    \tau_u\alphabar_{u0}(\tp1)}
     \stackrel{(b)}{=} 1 - \tau_u^{-1}\Ecal_u(\eta_u(t)), \label{eq:corruMmsePf}
\eeqa
where (a) follows from \eqref{eq:corruLim} and (b) follows from
\eqref{eq:alphaummse}.  Substituting in \eqref{eq:snru} into
\eqref{eq:corruMmsePf} proves the first equation in \eqref{eq:corrSEmmse}.
The second recursion in \eqref{eq:corrSEmmse} can be proven similarly.

For the initial condition in the recursion, observe that with $V(0) = \Exp[V_0]$,
the second order statistics are given by
\beqan
    \alphabar_{v0}(0) = \Exp[ V(0)^2] = (\Exp[V_0])^2, \quad
    \alphabar_{v1}(0) = \Exp[ V_0V(0)] = (\Exp[V_0])^2.
\eeqan
Hence, from \eqref{eq:corrvLim}, the initial correlation coefficient is
\[
    \rho_v(0) = \frac{\alphabar^2_{v1}(0)}{\tau_v\alphabar_{v0}(0)} =
        \frac{(\Exp[V_0])^2}{\tau_v},
\]
which agrees with the statement in the theorem.
This proves part (a).

To prove part (b), we again use Lemma \ref{lem:monoiter}.
Define the functions
\beq \label{eq:Huvmmse}
    H_u(\rho_v) := 1 - \tau_u^{-1}\Ecal_u(\beta \tau_v\rho_v/\tau_w),
    \quad
    H_v(\rho_u) := 1 - \tau_v^{-1}\Ecal_v(\tau_u\rho_u/\tau_w),
\eeq
and their concatenation
\beq \label{eq:Hmmse}
    H(\rho_v) = H_v(H_u(\rho_v)).
\eeq
From \eqref{eq:corrSEmmse}, it follows that $\rho_v(\tp1) = H(\rho_v(t))$.
Now, $\Ecal_u(\eta_u)$ and $\Ecal_v(\eta_v)$ defined in
\eqref{eq:mmseuv} are the mean-squared
errors of $U_0$ and $V_0$ under AWGN estimation measurements with SNRs
$\eta_u$ and $\eta_v$.  Therefore, $\Ecal_u(\eta_u)$ and $\Ecal_v(\eta_v)$
must be monotonically decreasing in $\eta_u$ and $\eta_v$.  Therefore,
$H_u(\rho_v)$ and $H_v(\rho_u)$ in \eqref{eq:Huvmmse}
are monotonically increasing functions and thus so is the concatenated
function $H(\rho_v)$ in \eqref{eq:Hmmse}.
Also, since the assumption of part (b) is that $\Ecal_u(\eta_u)$
and $\Ecal_v(\eta_v)$ are continuous, $H(\rho_v)$ is also continuous.
It follows from Lemma~\ref{lem:monoiter} that $\rho_v(t) \arr \rho_v^*$
where $\rho_v^*$ is a fixed point of \eqref{eq:corrSEmmse}.

It remains to show $\rho_v^* > 0$. Observe that
\beqa
    \Ecal_v(\eta_v) \stackrel{(a)}{=} \Exp[V_0 \MID Y=\sqrt{\eta_n}V_0+D]
     \leq  \var(V_0) \stackrel{(b)}{=}
    \Exp[V_0^2] - (\Exp[V_0])^2 = \tau_v(1-\rho_v(0)), \nonumber
\eeqa
where (a) follows from the definition of $\Ecal_v(\eta_v)$ in \eqref{eq:mmseuv}
and (b) follows from the definition of $\tau_v$ in \eqref{eq:tauuv} and
the initial condition $\rho_v(0) = (\Exp[V_0])^2/\tau_v$.
It follows from \eqref{eq:corrSEmmse} that
\[
    \rho_v(\tp1) = 1-\frac{1}{\tau_v}\Ecal_v(\eta_v(t)) \geq \rho_v(0).
\]
Therefore, the limit point $\rho_v^*$ of $\rho_v(t)$ must satisfy
$\rho_v^* \geq \rho_v(0) > 0$.

\subsection{Proof of Lemma \ref{lem:mmseZero}} \label{sec:mmseZeroPf}
Define the functions $H_u$, $H_v$ and $H$ as in \eqref{eq:Huvmmse} and
\eqref{eq:Hmmse} from the previous proof.   We know that $\rho_v(\tp1) = H(\rho_v(t))$.
When $\Exp[U_0] = \Exp[V_0] = 0$, the $\rho_v = 0$ is a fixed point of the update.
We can determine the stability of this fixed point by computing the derivative
of $H(\rho_v)$ at $\rho_v=0$.

Towards this end, we first use a
standard result (see, for example, \citep{GuoV:05}) that, for any prior on $U_0$
with bounded second moments, the mean-squared error in \eqref{eq:mmseuv} satisfies
\beq \label{eq:mmseuLow}
    \Ecal_u(\eta_u) = \frac{\tau_u}{1 + \eta_u\tau_u} + O(\eta_u^2).
\eeq
The term $\tau_u/(1+\eta_u\tau_u)$ is the minimum mean-squared error
with linear estimation of $U_0$ from an AWGN noise-corrupted measurement
$Y = \sqrt{\eta_u}U_0 + D$, $D \sim {\mathcal N}(0,1)$.
Equation \eqref{eq:mmseuLow} arises from the
fact that linear estimation is optimal in low SNRs -- see \citep{GuoV:05} for details.
Using \eqref{eq:mmseuLow}, we can compute the derivative of $H_u$,
\beqa
    H_u'(0) = -\frac{1}{\tau_u}\left. \frac{\partial}{\partial \rho_v}
        \Ecal_u(\beta \tau_v\rho_v/\tau_w) \right|_{\rho_v = 0}
        = -\frac{\beta \tau_v}{\tau_u\tau_w}\Ecal_u'(0)
        = \frac{\beta \tau_u \tau_v}{\tau_w} \label{eq:mmseuDeriv}
\eeqa
Similarly one can show that
\beq \label{eq:mmsevDeriv}
    H_v'(0) = \frac{\tau_u \tau_v}{\tau_w},
\eeq
and hence
\beq \label{eq:HzeroDeriv}
    H'(0) = H_v'(0)H_u'(0) = \frac{\beta \tau_u^2 \tau_v^2}{\tau_w^2}.
\eeq

We now apply a standard linearization analysis of the nonlinear system
$\rho_v(\tp1) = H(\rho_v(t))$ around the fixed point $\rho_v=0$.
See, for example, \citep{Vidyasagar:78}.
If $\sqrt{\beta}\tau_u\tau_v < \tau_w$ then $H'(0) < 1$ and the fixed point
is stable.  Thus, for any $\rho_v(0)$ sufficiently small $\rho_v(t) \arr 0$.
This proves part (b) of the lemma.

On the other hand, if $\sqrt{\beta}\tau_u\tau_v > \tau_w$ then $H'(0) > 1$
and the fixed point is unstable.  This will imply that for any $\rho_v(0) > 0$,
$\rho_v(t)$ will diverge from zero.  But, we know from Theorem \ref{thm:SEmmse}
that $\rho_v(t)$ must converge to some fixed point of $\rho_v = H(\rho_v)$.
So, the limit point must be positive.  This proves part (a) of the lemma.

\bibliographystyle{imaiai}

\newcommand{\SortNoop}[1]{}
\ifx\undefined\BySame
\newcommand{\BySame}{\leavevmode\rule[.5ex]{3em}{.5pt}\ }
\fi
\ifx\undefined\textsc
\newcommand{\textsc}[1]{{\sc #1}}
\newcommand{\emph}[1]{{\em #1\/}}
\let\tmpsmall\small
\renewcommand{\small}{\tmpsmall\sc}
\fi

\end{document}